\documentclass[12pt,letterpaper]{article}

\usepackage[letterpaper,margin=1.25in]{geometry}

\usepackage{setspace}
\usepackage{float}

\usepackage{amsmath, amssymb, amsthm}
\usepackage{graphicx}
\usepackage{booktabs}
\usepackage[unicode=true,
 bookmarks=false,
 breaklinks=false,pdfborder={0 0 0},pdfborderstyle={},backref=false,colorlinks=true]
 {hyperref}
\hypersetup{citecolor=black, linkcolor=black, urlcolor=black}
\usepackage{bbm}
\usepackage{framed}

\usepackage{tikz}
\usepackage{caption, subcaption}
\usetikzlibrary{patterns}
\usepackage{natbib}

\theoremstyle{definition}

\newtheorem{lemma}{Lemma}

\newtheorem{corollary}{Corollary}

\theoremstyle{definition}
\newtheorem{definition}{Definition}
\newtheorem*{definition*}{Definition} 

\newtheorem{remark}{Remark}

\newtheorem*{lemma*}{Lemma}

\newtheorem{proposition}{Proposition}
\newtheorem{Theorem}{Theorem}

\usepackage{enumitem}

\theoremstyle{definition}

\usepackage{etoolbox}
\newtheorem{example}{Example}
\AtEndEnvironment{example}{\hfill$\square$}

\newcommand{\definitionname}{Definition}

\newtheoremstyle{nopunct}
  {3pt}{3pt}
  {\itshape}
  {}
  {\bfseries}
  {}
  {0.5em}
  {\thmname{#1}\thmnote{ #3}}
\theoremstyle{nopunct}
\newtheorem*{defcontx}{} 

\title{Robustness of Persuasion to Receiver Preferences\thanks{Gradwohl gratefully acknowledges support from the Israel Science Foundation, Grant 2039/24. Hu and Smorodinsky gratefully acknowledge support from the German Research Foundation, Grant 514505843 and the Israel Science Foundation, Grant 2646/22. Hu additionally acknowledges the generous financial support from Israel Planning and Budgeting Committee (VATAT). We thank Itai Arieli, Yakov Babichenko and Doron Ravid for helpful comments.}}

\author{
Ronen Gradwohl\thanks{Department of Economics, University of Haifa, Israel. \texttt{rgradwohl@econ.haifa.ac.il  }}
\and
Fengming Hu\thanks{Faculty of Data and Decision Sciences, Technion, Israel. \texttt{fengming.hu@campus.technion.ac.il}}
\and
Rann Smorodinsky\thanks{Faculty of Data and Decision Sciences, Technion, Israel. \texttt{rann@technion.ac.il}}
}

\date{} 

\begin{document}


\maketitle

\begin{abstract}
We study the robustness of Bayesian persuasion to (Knightian) uncertainty about the receiver’s preferences. We consider two properties:  {\em continuity}, in which only the modeler lacks precise knowledge while the model's predictions are nonetheless accurate; and {\em robustness}, in which the sender also lacks precise knowledge, but where the outcome is insensitive to this ignorance.
For the latter, we model the sender’s behavior as either maximizing worst-case utility or minimizing the regret.
We show that the notions are equivalent. In addition, they are equivalent to robustness to tie-breaking. Finally, we show that robustness holds generically.
\end{abstract}

\section{Introduction}
Bayesian persuasion, introduced by \citet*{Kamenica2011BayesianPersuasion}, is one of the most widely studied models in microeconomics. This is likely due to its simplicity and its applicability to various domains in economics, law, political science and more \citep*{Kamenica2019BayesianDesign}. In its basic form, there are two agents: one agent, the `sender', has access to information about a state of the world, but takes no action. Instead, the sender shares her information, partially and strategically, with a second agent, the `receiver'. The receiver then takes an action that determines the utilities of both agents. A main assumption in Bayesian persuasion is that the sender can devise and commit to her information-sharing procedure before she observes the realized state. The central observation, dating back to work by \citet*{aumann1995repeated} on repeated zero-sum games with incomplete information on one side, is that the maximal utility of the sender is equal to the value of the concavification of her indirect utility function---the implied utility function given that the receiver best-replies to his updated belief. 

How robust is this observation, and to what extent does it depend on strong modeling assumptions? In the so-called Wilson Doctrine, \citet*{wilson1987rgame}
argues that 
{\it ``Game theory has a great advantage in explicitly analyzing the consequences of trading rules that presumably are really common knowledge; it is deficient to the extent it assumes other features to be common knowledge, such as one agent's probability assessment about another's preferences or information.''} 
However, the Bayesian persuasion model does assume that the receiver's preferences are common knowledge.
In this paper, we ask whether this assumption is crucial for the model's main observation to hold, or whether the model is robust.

More specifically, we consider two conceptually distinct notions of robustness. The first notion, which we call {\em continuity}, supposes that the sender and receiver do have complete knowledge of all features of their interaction, but that the analyst or modeler lacks some knowledge. The model is {\em continuous} if, despite the lack of complete knowledge, the modeler can nonetheless accurately predict the outcome of the interaction \citep*{kajii1997robustness}.

The second notion, which we henceforth refer to as {\em robustness}, is more along the lines of the Wilson Doctrine, and supposes that the sender also does not have complete knowledge. The model is {\em robust} if, despite the sender's lack of complete knowledge, the outcome of the interaction is nonetheless close to the benchmark in which all features are common knowledge.

Formally, we ask whether the concavification result is continuous and robust when the precise preferences of the receiver are no longer common knowledge. We assume that the lack of common knowledge manifests as infinitesimal ignorance. Specifically, for every action and state, the receiver's utility is only known to be within some small interval, without any distributional assumptions (e.g., ambiguity or Knightian uncertainty). The Bayesian persuasion model is continuous if the maximal utility attained by the sender---the concavification value---can be (approximately) predicted by the modeler even if the modeler only knows the receiver's utilities within such small intervals.

To study robustness, we first elaborate on the sender's strategy in light of her ignorance.
To model this strategy, we invoke two paradigms from decision theory \citep*{savage1951theory}: max-min utility and regret minimization. A max-min sender chooses a strategy that maximizes her worst-case utility. A regret-minimizing sender chooses a strategy that minimizes her worst-case regret, where regret is the difference in utility between the actual decision made and what would have been the optimal decision in hindsight. Under these two paradigms, the Bayesian persuasion model is then robust if the maximal utility attained by the sender---the concavification value---is also (close to) what she obtains when she only knows the receiver's utilities within small intervals, and when she chooses her strategy either to maximize worst-case utility or to minimize regret.

We provide an example that motivates our model and results. 

\begin{example}

\label{sec:warm_examples}
A candidate applying for a position that requires substantial expertise (e.g., an academic position) is either qualified or unqualified;  the state space is $\Omega=\{0,1\}$, and the prior is $\mu_0=0.5$.
The candidate (the sender, she) controls her reputation by suggesting whom to approach for letters of reference (Bayesian persuasion).
The recruiter (the receiver, he) has three options: reject the candidate ($a_1$), offer a short-term contract ($a_2$), or offer a long-term contract ($a_3$).

If the candidate is unqualified, the receiver prefers to reject her; of the two remaining options, he prefers to offer a short-term contract. If the candidate is qualified, rejection is the receiver’s worst option, but he is almost indifferent between a short-term and a long-term contract.
Receiver and sender utilities, denoted by $u$ and $v$, respectively, are specified in Figure~\ref{table:uvalue_example1_intro}. In particular, in state $1$, the receiver's utility for action $a_3$ lies within an interval of length $\delta\ll 1$ around his utility for action $a_2$.

\begin{figure}[h!]
  \centering
  \begin{subfigure}[b]{0.5\linewidth}
    \begin{tabular}{c|c|c|c}
\hline
$u(~,~)$ & \textbf{$a_1$} & { \textbf{$a_2$}}& \textbf{$a_3$}\\ \hline
$\omega=0$ & 0.84 & $ 0.48 $& 0.05   \\ \hline
${\omega}=1$ & 0.72 & 0.84 & $0.84\pm \tfrac{\delta}{2}$   \\ \hline
\multicolumn{1}{p{2cm}}{}\\ \hline 
$v(~,~)$ & {$a_1$} & { \textbf{$a_2$}}& \textbf{$a_3$}\\ \hline
$\omega=0,1$ & 0 & 0.1 & 1  \\ \hline
\multicolumn{1}{p{2cm}}{}\\
\multicolumn{1}{p{2cm}}{}
\end{tabular}
    \caption{receiver and sender utilities}
    \label{table:uvalue_example1_intro}
  \end{subfigure}\hfill
 \begin{subfigure}[b]{0.5\textwidth}
      \tikzset{every picture/.style={inner sep=0, outer sep=0}}
    \begin{tikzpicture}[ scale = 0.35]
         \draw[white] (-2,0)--(-0.5,0);
  \node (a11) at (0,8.4)[label={[left=1pt]\textcolor{black}{$a_1$}}]{};
\node (a21) at (0,4.8) [label={[left=1pt]\textcolor{blue}{$a_2$}}]{};
\node (a31) at (0,0.5) [label={[left=1pt]\textcolor{red}{$a_3$}}]{};

\node (zero) at (0,0) [label={[below=4pt]$0$}]{};
\node (zero) at (12,0) [label={[below=4pt]$0.75$}]{};
\node (zero) at (8,0) [label={[below=4pt]$\mu_0$}]{};
\node (zero) at (16,0) [label={[below=4pt]$1$}]{};

  \node (a12) at (16,7.2){};
\node (a22) at (16,8.4) {};
\node (a32) at (16,8.4) {};

    \draw[black]  (a11) -- (a12);
     \draw[black]  (a21) -- (a22);
     \draw[black]  (a31) -- (a32);

  \draw[->,black]  (0,0) -- (0,11);
     \draw[->,black]  (0,0) -- (18,0);

  \draw[<->,dashed, thick,black]  (0,10.5) -- (12,10.5);
  
  \draw[<->,dashed, thick, blue]  (12,10.5) -- (16,10.5);

\filldraw[red] (16,8.5) circle (6pt);

   \draw[dashed,gray]  (12,0) -- (12,11);
     \draw[dashed,black]  (16,0) -- (16,11);             
\node[] at (6,-2) {receiver's belief};      
   
    \end{tikzpicture}
    \caption{best-reply correspondences}
    \label{fig:vvalue_example1_intro}
    \end{subfigure}
  \caption{Example 1.}
  \label{fig:uvalue_example_full_page2}
\end{figure}

Can we neglect $\delta$ when it is very small? In other words, how accurate are our predictions if we simply set  $\delta$ to zero? 
Note that when $\delta$ is zero, $\mu=1$ is the unique belief at which $a_3$ is a best reply for the receiver (see Figure~\ref{fig:vvalue_example1_intro}). 
Under $\delta=0$, full revelation yields an expected utility for the sender of $\mu_0 \cdot 1 + (1-\mu_0)\cdot 0=0.5$ under sender-preferred tie-breaking.

We consider the following three cases for the sender:
\begin{enumerate}
    \item She knows the receiver's utility and maximizes her expected payoff.
    \item She does not know the receiver's utility, and she is \emph{max-min}: she chooses a signal policy that maximizes her worst-case expected utility.
    \item She does not know the receiver's utility, and she is \emph{regret-minimizing}: she chooses a signal policy that minimizes her worst-case regret. 
\end{enumerate}

How good are our predictions for these three situations? The first corresponds to continuity, and the latter two correspond to the two forms of robustness. 

This example is neither continuous nor robust under either behavioral criterion. 
For arbitrary $\delta>0$, we can find a receiver utility that satisfies $u(a_3,1)<0.84$.
However, if $u(a_3,1)<0.84$, the sender's most preferred action $a_3$ is never a best reply for the receiver, regardless of the latter's posterior belief. 
In this case, the corresponding optimal signal policy yields an expected utility below $v(a_2)=0.1$. This value is bounded away from the expected utility at $\delta=0$, so continuity fails.

We now consider robustness. 
In this paper, we do not characterize the max-min and regret-minimizing policies explicitly.\footnote{See \citet{Babichenko2022Regret-minimizingPersuasion,sapiro2024persuasion,parakhonyak2026persuasion} for further details on max-min and regret-minimizing policies in binary-action settings.}
Instead, we explore the connections among these notions.
Note that neither the max-min policy nor the regret-minimizing policy can give the sender a higher payoff than she could obtain by implementing the optimal policy for each receiver.
In this example, consider a receiver whose utility satisfies $u(a_3,1)<0.84$.
Even the optimal policy for this receiver gives the sender an expected utility below $v(a_2)=0.1$, so her expected utility under either behavioral criterion is also below $0.1$.
This is far from the prediction of $0.5$ for the case $\delta=0$. Hence, robustness fails under both behavioral criteria.
The implication illustrated here---that failure of continuity implies failure of robustness---is stated formally in Lemma~\ref{prop:continuous_robust_nece}.

\begin{figure}[h]
    \centering
      \begin{subfigure}[b]{0.35\textwidth}
      \tikzset{every picture/.style={inner sep=0, outer sep=0}}
    \begin{tikzpicture}[baseline = 3cm, scale = 0.3]
    
  \node (va2) at (16,10){};
\node (va41) at (12,2) {} ;
\node (va42) at (16,2) {} ;

\node (zero) at (0,0) [label={[below=4pt]$0$}]{};
\node (zero) at (12,0) [label={[below=4pt]$0.75$}]{};
\node (zero) at (16,0) [label={[below=4pt]$1$}]{};
\node (zero) at (8,0) [label={[below=4pt]$\mu_0$}]{};

  \draw[->,black]  (0,0) -- (0,11);
     \draw[->,black]  (0,0) -- (18,0);

    \filldraw[red] (16,10) circle (3pt);
   \draw[thick,red]  (16,0) -- (16,10);
   
  \draw[thick,blue]  (va41) -- (va42);
     \draw[dashed,gray]  (12,0) -- (va41);
      \draw[dashed,gray]  (16,0) -- (16,5);
      
          \draw[very thick,black]  (va2)  --  (0,0);
         \draw[dashed,gray]  (0,0)  --  (va41);
          \draw[<->,very thick,blue]  (8,1)  --  (8,5);
     
\node[] at (6,-2) {receiver's belief};      
        
    \end{tikzpicture}
    \caption{$\delta=0$}
    \label{fig:inconsistent_case_regret_1}

    \end{subfigure}
    \hspace{3cm}
       \begin{subfigure}[b]{0.35\textwidth}
      \tikzset{every picture/.style={inner sep=0, outer sep=0}}
    \begin{tikzpicture}[baseline = 3cm, scale = 0.3]
   
  \node (va2) at (16,10){};
\node (va41) at (12,2) {} ;
\node (va42) at (16,2) {} ;

\node (zero) at (0,0) [label={[below=4pt]$0$}]{};
\node (zero) at (16,0) [label={[below=4pt]$1$}]{};
\node (zero) at (12,0) [label={[below=4pt]$0.75$}]{};
\node (zero) at (8,0) [label={[below=4pt]$\mu_0$}]{};

  \draw[->,black]  (0,0) -- (0,11);
     \draw[->,black]  (0,0) -- (18,0);
  \draw[thick,blue]  (va41) -- (va42);
     \draw[dashed,gray]  (12,0) -- (va41);
      \draw[dashed,gray]  (16,0) -- (16,5);
      
          \draw[dashed,gray]  (16,2)  --  (0,0);
         \draw[very thick,black]  (0,0)  --  (va41);
          \draw[<->,very thick,red]  (8,0.75)  --  (8,1.75);
     
\node[] at (6,-2) {receiver's belief};      
        
    \end{tikzpicture}
    \caption{$u(a_3,1)<0.84$}
    \label{fig:inconsistent_case_regret_2}

    \end{subfigure}
    \caption{Sender's indirect utilities. The vertical arrows represent regrets.}
    \label{fig:inconsistent_case}
\end{figure}

We further use this example to show the distinction between two behavioral criteria.
The max-min sender is concerned that $u(a_3,1)<0.84$, in which case full revelation reduces her chance of receiving the short-term contract (shown by the dashed line in Figure~\ref{fig:inconsistent_case_regret_2}). 
It is straightforward  that the max-min signal policy splits the prior between the posterior beliefs $0$ and $0.75$.
This policy performs poorly under the regret-minimization criterion: its maximal regret (the blue arrow in Figure~\ref{fig:inconsistent_case_regret_1}), attained when $a_3$ is viable, is at least $0.4$. 
By contrast, whenever $a_3$ is viable, the full-revelation signal policy, which need not itself be regret-minimizing, induces the receiver to choose $a_3$ with probability $0.5$. The maximal regret of full revelation is attained when $a_3$ is not viable and is less than $0.1$ (the red arrow in Figure~\ref{fig:inconsistent_case_regret_2}).
Thus, the max-min policy is not regret-minimizing.
\end{example}

The preceding example illustrates that with respect to receiver preferences, continuity and robustness need not hold.
In addition, observe that the prediction of 0.5 for $\delta=0$ relies on sender-optimal tie-breaking for the receiver, and that a different tie-breaking rule would lead to a different prediction.
Hence, the example is also not robust to tie-breaking. Although these notions are conceptually distinct, their simultaneous failure suggests that they are related.

The main result of this paper is that all four notions are equivalent in every Bayesian persuasion model with finite states and finite actions: continuity with respect to receiver preferences, max-min robustness and regret-minimizing robustness to uncertainty about those preferences, and robustness to tie-breaking. Our primary focus is on the first three. The equivalence follows because all four notions have the same necessary and sufficient conditions. 

In Section~\ref{sec:model}, we model uncertainty about receiver preferences as infinitesimally small, non-probabilistic (Knightian) uncertainty. Proposition~\ref{lemma:exis_max_regret} shows that max-min and regret-minimizing policies always exist in Bayesian persuasion models with finite state and action spaces and uncertain receiver utilities. We then define continuity and the two forms of robustness in Definitions~\ref{def:def_continuous_BP}, \ref{def:def_robust_BP_1}, and~\ref{def:def_robust_BP_2}.

Section~\ref{sec:main_equivelant} establishes the equivalence among the first three notions (Theorem~\ref{theorem:bp_equivalent}). All three hold if and only if there exists an optimal policy supported on non-critical posterior beliefs (Definition~\ref{def:critical_mu}). 
The sufficiency and necessity directions are proved in Sections~\ref{sec:sufficient}~and~\ref{sec:nece_example} (Propositions~\ref{prop:suffi_all}~and~\ref{prop:necessary_all}), respectively.
Proposition~\ref{lemma:set_upper_continuous} establishes an unconditional upper-semicontinuity result, which is used in both proofs and also yields a weaker prediction that always holds.

In Section~\ref{sec:tie_eq}, we define robustness to tie-breaking (Definition~\ref{defi:robust_tie_breaking}) and establish its equivalence with the first three notions (Theorem~\ref{thy:equive_tie_breaking}). 
The proof shows that robustness to tie-breaking has the same necessary and sufficient conditions as the first three notions (Proposition~\ref{prop:nece_suf_tie_breaking}). 
We then provide a numerical necessary and sufficient condition for all four notions (Proposition~\ref{prop:num_nece_suff}). 
This condition corresponds to a special tie-breaking rule (Definition~\ref{def:special_tie}) and is closely related to the techniques used in the necessity proofs in Section~\ref{sec:nece_example}. Combining this equivalence with the genericity result in \cite{lipnowski2024perfect}, we show that continuity and robustness hold generically (Corollary~\ref{thy:generic_robust}). Finally, Section~\ref{sec:final_discuss} provides further examples of non-robustness.

\subsection{Related literature}
\label{sec:lit}

The motivation for this work can be traced back to the literature on the robustness of equilibria to incomplete information \citep*{kajii1997robustness,ui2001robust,morris2005generalized}. 
In these papers, an incomplete-information game is close to a benchmark complete-information game when its payoffs coincide with the benchmark payoffs with high probability. They ask whether a given benchmark equilibrium  (or a given set of equilibria) can be approximated by an equilibrium of every sufficiently close incomplete-information game.
Related work asks whether predictions of the original game remain valid across unknown private-information structures or information-acquisition technologies \citep{bergemann2013robust,bergemann2017belief,denti2023robust}.
Although our work is not as general, it diverges from the above literature since these papers maintain a Bayesian framework, whereas we study robustness to non-probabilistic uncertainty.

Our focus is on Bayesian persuasion \citep*{aumann1995repeated,Kamenica2011BayesianPersuasion,Kamenica2019BayesianDesign}, where we study continuity and robustness. 
Within the large literature on Bayesian persuasion, two lines of research are particularly relevant to our work. One studies optimal signal policies in environments that depart from the benchmark model. The other asks whether the predictions of the benchmark Bayesian persuasion model remain valid under perturbations or uncertainties. Our work contributes to the second line.

Numerous papers focus on the first. 
\citet*{Dworczak2022PreparingPersuasion} characterize the optimal policy when the sender is uncertain about the exogenous sources of information the receiver may learn from and about strategy selection. 
In settings with binary receiver actions, several papers study optimal signal policies under various forms of uncertainty using different criteria  \citep*{Kosterina2022PersuasionBeliefs,Babichenko2022Regret-minimizingPersuasion, Feng2024Rationality-RobustResponse,sapiro2024persuasion,parakhonyak2026persuasion}.
This line of research naturally extends to algorithmic questions. For example, \citet*{Zu2021LearningIgnorance,Castiglioni2023RegretModels} study online  persuasion algorithms under uncertainty.
Beyond the distinctions in research direction, the policies developed in these papers do not apply to our results, because they either address substantially different environments or restrict the receiver to binary actions.

Several papers contribute to both lines of research. \citet{Hu2021RobustReceiver} and \citet{MATYSKOVA2023costly} study settings where the receiver has private information. Both papers characterize optimal policies and show that the sender's optimal expected utility approaches the benchmark Bayesian-persuasion value as the private information becomes uninformative. 
\citet*{lipnowski2022persuasion} characterize equilibria and optimal policies when the sender has imperfect commitment. They also show that, generically, as the sender's credibility approaches one, she can secure expected utility arbitrarily close to the full-commitment benchmark.
Unlike these papers, we study a different deviation from the benchmark, namely, uncertainty about receiver preferences. 

Within the second line of research, \citet*{lipnowski2024perfect} study the robustness of Bayesian persuasion to the receiver's tie-breaking rule. 
Their central result is that Bayesian persuasion is generically robust to tie-breaking. 
Our central result instead establishes the equivalence of four distinct notions with a focus on robustness to receiver preferences.
Underlying their genericity result is a condition that holds generically and is sufficient, but not necessary. By contrast, underlying our equivalence results is an interpretable necessary and sufficient structural characterization that is not provided in their work. Our characterization also complements their analysis of tie-breaking.

More broadly, \citet{von2010leadership} study robustness with respect to tie-breaking in Stackelberg games, while \citet{gan2025robust, lin2025generalized} study robustness to approximately optimal follower responses in Stackelberg generalized principal-agent games.
The latter two rely on positive inducibility gaps to obtain their robustness results.
These results do not shed light on our model as the inducibility gap is not positive.

Finally, a parallel line of literature studies robustness to model assumptions in cheap talk games \citep*{diehl2021non,dilme2022robust,arieli2023informationally,steg2023robust}. 

\section{Model and Preliminaries}
\label{sec:model}
A Bayesian persuasion model (BP model) is a tuple $G=\{\Omega, A, \mu_0, u, v\}$.
$\Omega$ is a finite, non-singleton set of states, and the state is realized according to a full-support prior $\mu_0 \in \Delta (\Omega)$.
$A$ is a finite, non-singleton set of actions available to the receiver.
$u: A\times \Omega \rightarrow \mathbb{R}$ and $v: A\times \Omega \rightarrow \mathbb{R}$  denote the utility functions of the receiver and sender, respectively.
We assume that the sender's utilities are bounded and nonnegative, and normalize them so that $v: A\times \Omega \rightarrow [0,1]$.\footnote{Since our analysis focuses on changes in the receiver’s utilities, we impose a uniform bound on the sender’s utilities for simplicity, without loss of generality.}

The game proceeds as follows. The sender has a message space and devises a signal policy mapping each state to a (possibly random) message.
Then the state is realized and a message is publicly announced according to that signal policy, after which the receiver takes an action and both utilities are realized. 

It follows from the celebrated Aumann-Maschler splitting lemma \citep*{aumann1995repeated} that a signal policy is equivalent to a probability measure over posterior beliefs whose average is equal to the prior.
We denote the set of all signal policies by
\begin{equation}
  \Pi :=\left\{\pi\in \Delta(\Delta(\Omega)): \int_{\Delta(\Omega)} \mu \pi(d\mu)=\mu_0 \right\} ~.
  \label{eq:signal_policy}
\end{equation}
A posterior belief $\mu$  belongs to the support of $\pi$, $\mu\in\operatorname{supp}(\pi)$, if every open neighborhood $Q$ of $\mu$ satisfies $\pi(Q)>0$.
We endow $\Delta(\Delta(\Omega))$ with the weak topology.
Since $\Omega$ is finite, $\Delta(\Omega)$ is a compact metric space and, in particular, is complete and separable.
Consequently, $\Delta(\Delta(\Omega))$ is compact under the weak topology by Prokhorov's theorem.\footnote{A sequence of probability measures $\{\pi_n\}_{n\in \mathbb{N}^+}$ weakly converges to $\pi$ if for every bounded and continuous function $f$, $\mathbb{E}_{\pi_n}[f]\rightarrow \mathbb{E}_{\pi}[f]$. Since $\Delta(\Omega)$ is a compact metric space, the weak topology on $\Delta(\Delta(\Omega))$  is metrizable. Therefore, we use ``sequences'' instead of ``nets'' \citep{billingsley2013convergence}.  } 
Appendix~\ref{app:mode_proofs} proves the following observation.
\begin{lemma}
\label{lemma:Pi_compact}
    $\Pi$ is compact under the weak topology.
\end{lemma}
We now extend the standard BP model to allow the receiver’s payoff function to vary across payoff types.
A BP model with a receiver of type $\theta$ is denoted by $G(\theta)=\{\Omega, A, \mu_0, u_{\theta}, v\}$, where the receiver's utility profile is uniquely determined by the type $\theta:=(u_{\theta}(a,\omega))_{(a,\omega)\in A\times \Omega}\in \mathbb{R}^{|\Omega|\times |A|}$. 
We endow the type space with the sup norm. 
We say that a sequence of BP models
$\{G(\theta_n)\}_{n\in\mathbb{N}^+}$ converges to $G(\theta_0)$ in receiver utilities if $\displaystyle\lim_{n\rightarrow\infty}\theta_n=\theta_0$.  
For any $\delta>0$, we denote the closed $\delta$-neighborhood of a type $\theta$ by $\operatorname{Ball}_{\delta}(\theta)$.

We denote the receiver's best-reply correspondence by $A^*:\mathbb{R}^{|\Omega|\times |A|}\times\Delta(\Omega)\rightrightarrows A$,
where for each $(\theta,\mu)$
\begin{equation}
A^*(\theta,\mu):=\{a\in A: \sum_{\omega\in \Omega} u_{\theta}(a,\omega) \mu(\omega)  \geq \sum_{\omega\in \Omega} u_{\theta}(b,\omega) \mu(\omega) \quad \forall b\in A  \}~.
\end{equation}
We denote the receiver's best-reply region by $B:\mathbb{R}^{|\Omega|\times |A|}\times A\rightrightarrows\Delta(\Omega)$, where for each $(\theta,a)$,
\begin{equation}
B(\theta, a):= \{\mu \in \Delta(\Omega): \  a \in A^*(\theta,\mu)\}~.  
\end{equation}
For each $(\theta,a)$, $B(\theta,a)$ is a compact and convex subset of $\Delta(\Omega)$, possibly empty. 

In the lion's share of the literature, the focus is on receivers who break ties in favor of the sender (e.g., sender-preferred tie-breaking). Define the sender's indirect payoff function $V:\mathbb{R}^{|\Omega|\times |A|}\times\Delta(\Omega)\to [0,1]$ as follows: for each $(\theta,\mu)$,
\begin{equation}
V(\theta, \mu)=\max_{a\in A^*(\theta,\,\mu)}\sum_{\omega\in\Omega}\mu(\omega)v(a,\omega)~.
\end{equation}
Throughout,  $V$ denotes the indirect utility function under the sender-preferred tie-breaking.
At posterior belief $\mu$ and receiver type $\theta$, we call action $a$ an 
\emph{induced action} under sender-preferred tie-breaking if
\begin{equation*}
    a\in  A^*(\theta,\mu)\quad \text{and}\quad \sum_{\omega\in\Omega}\mu(\omega)v(a,\omega)=V(\theta, \mu)~.
\end{equation*}

\begin{lemma}
    \label{lemma:V_upper_semi}
$V:\mathbb{R}^{|\Omega|\times |A|}\times\Delta(\Omega)\to[0,1]$ is upper semicontinuous.
\end{lemma}
The proof is provided in Appendix~\ref{app:mode_proofs}.
Define the sender's expected utility $W:\mathbb{R}^{|\Omega|\times |A|}\times \Pi\to[0,1]$ as follows:  for every receiver type $\theta$ and signal policy $\pi$,
\begin{equation}
W(\theta,\pi)=\int_{\Delta(\Omega)} V(\theta,\mu)\pi(d\mu)~.
\end{equation}

\begin{lemma}
    \label{lemma:joint_upper_W}
$W:\mathbb{R}^{|\Omega|\times |A|}\times \Pi\rightarrow [0,1]$ is upper semicontinuous.
\end{lemma}
The proof is provided in Appendix~\ref{app:mode_proofs}.
For each receiver type, $\pi\mapsto W(\theta,\pi)$ is upper semicontinuous on the compact set $\Pi$; hence, it attains a maximum and an optimal policy exists.
We now define the notion of continuity of BP models.
\begin{definition}
\label{def:def_continuous_BP}
A BP model $G(\theta_0)$ is  \textbf{continuous} if for any sequence $\{G(\theta_n)\}_{n\in\mathbb{N}^+ }$ with $\theta_n\rightarrow\theta_0$,  $\displaystyle \lim_{n\rightarrow \infty}\, \left[\max_{\pi\in\Pi}W(\theta_n,\pi)\right] =\max_{\pi\in\Pi}W(\theta_0,\pi)$.
\end{definition}

Continuity of a BP model implies that slight ignorance by the modeler with respect to the receiver's type does not substantially change the modeler's prediction about the sender's expected utility.

\subsection{Robustness}
Continuity captures the setting where the modeler is (infinitesimally) ignorant about receiver preferences. We now turn to the setting where, in addition, the sender is (infinitesimally) ignorant. We depart from the Bayesian framework to model such ignorant senders and assume that an ignorant sender has no probabilistic model for the receiver's type. Rather, she assumes that the type is in a restricted set of types, typically a small ball.  
Thus, we replace the standard assumption that a sender maximizes expected utility with one of two non-Bayesian decision-making schemes: either maximizing  worst-case utility (\emph{max-min}) or minimizing regret (\emph{regret-minimizing}).

For any collection of types $\Theta\subseteq \mathbb{R}^{|\Omega|\times|A|}$ and any signal policy $\pi$, $\inf_{\theta\in\Theta} W(\theta,\pi)$ is the worst-case expected utility of $\pi$, while $\sup_{\theta\in\Theta} \left [  \max_{\pi'\in \Pi}W(\theta,\pi')- W(\theta,\pi) \right]$ is the worst-case regret of $\pi$.
A max-min sender aims to maximize the former, whereas a regret-minimizing  sender aims to minimize the latter. We now establish the existence of two signal policies that attain the respective objectives.

\begin{proposition}
\label{lemma:exis_max_regret}
For any nonempty collection of types $\Theta\subseteq \mathbb{R}^{|\Omega|\times|A|}$, there exist two signal policies $\pi_1$ and $\pi_2$ satisfying
 \begin{equation*}
     \begin{cases}
       \inf_{\theta\in\Theta} W(\theta,\pi_1)=\sup_{\pi\in \Pi} \inf_{\theta\in\Theta} W(\theta,\pi)~,
       \\
        \sup_{\theta\in\Theta} \left[\max_{\pi'\in \Pi}W(\theta,\pi') -W(\theta,\pi_2) \right] =\inf_{\pi\in \Pi} \sup_{\theta\in\Theta} \left [  \max_{\pi'\in \Pi}W(\theta,\pi')- W(\theta,\pi) \right]~.
     \end{cases}
 \end{equation*}
 We call $\pi_1$ the \textbf{max-min} policy and $\pi_2$ the \textbf{regret-minimizing} policy.    
\end{proposition}
\begin{proof}
It suffices to show the existence of $\pi_2$ such that 
\begin{equation*}
    \inf_{\theta\in\Theta} \left[W(\theta,\pi_2)-\max_{\pi'\in \Pi}W(\theta,\pi')  \right] =\sup_{\pi\in \Pi} \inf_{\theta\in\Theta} \left [   W(\theta,\pi)-\max_{\pi'\in \Pi}W(\theta,\pi') \right]~.
\end{equation*}
For each fixed type $\theta$,    $W(\theta,\pi)$ and $ W(\theta,\pi)-\max_{\pi'\in \Pi}W(\theta,\pi')$ are upper semicontinuous on $\Pi$. Thus, $\inf_{\theta\in\Theta} W(\theta,\pi)$ and $\inf_{\theta\in\Theta} \left [   W(\theta,\pi)-\max_{\pi'\in \Pi}W(\theta,\pi') \right]$ are upper semicontinuous on $\Pi$, because the infimum of any collection of upper semicontinuous functions is upper semicontinuous. Because $\Pi$ is compact, these two upper semicontinuous functions attain their maxima. 
\end{proof}

We now define the notion of robustness of BP models under both criteria. For a sender that maximizes the minimal utility:
\begin{definition}
\label{def:def_robust_BP_1}
A BP model $G(\theta_0)$ is  \textbf{MM-robust} if for any $\epsilon>0$, there exists $\delta>0$ such that for any nonempty set of types $\Theta\subseteq \operatorname{Ball}_{\delta}(\theta_0)$ with $\theta_0\in\Theta$, every \textbf{max-min} policy $\pi^*$ satisfies 
$ \sup_{\theta\in \Theta} \left|W(\theta,\pi^*)-\max_{\pi\in\Pi} W(\theta_0,\pi)\right|<\epsilon$.
\end{definition}
For a sender that minimizes the maximal regret: 
\begin{definition}
\label{def:def_robust_BP_2}
A BP model $G(\theta_0)$ is  \textbf{RM-robust} if for any $\epsilon>0$, there exists $\delta>0$ such that for any nonempty set of types $\Theta\subseteq \operatorname{Ball}_{\delta}(\theta_0)$ with $\theta_0\in\Theta$, every \textbf{regret-minimizing} policy $\pi^*$ satisfies $ \sup_{\theta\in \Theta} \left|W(\theta,\pi^*)-\max_{\pi\in\Pi} W(\theta_0,\pi)\right|<\epsilon$.
\end{definition}

\section{Continuity and Robustness Are Equivalent}
\label{sec:main_equivelant}
In this section, we establish the equivalence of the three notions. Formally,
\begin{Theorem} \label{theorem:bp_equivalent}
For any BP model $G(\theta_0)$, the following statements are equivalent:
\begin{itemize}
    \item  $G(\theta_0)$ is continuous. 
    \item  $G(\theta_0)$ is MM-robust.
    \item  $G(\theta_0)$ is RM-robust.
\end{itemize}
\end{Theorem}

In the following two subsections, we characterize necessary and sufficient conditions for continuity, MM-robustness, and RM-robustness (Propositions~\ref{prop:suffi_all}~and~\ref{prop:necessary_all}). 
The three notions admit the same characterization. Therefore, they are equivalent.
The common condition---the existence of an optimal policy supported on non-critical posterior beliefs (see 
Definition~\ref{def:critical_mu})---is readily interpretable and is stated formally at the end of Section~\ref{sec:nece_example}.

We first show that sufficiently small ignorance always permits a weaker prediction, namely, the identification of a valid upper bound on the outcome of arbitrary BP instances.

\begin{proposition}
\label{lemma:set_upper_continuous}
For any BP model $G(\theta_0)$ and every $\epsilon>0$, there exists $\delta>0$ such that, for every nonempty set of types $\Theta\subseteq\operatorname{Ball}_{\delta}(\theta_0)$, every max-min policy $\pi_1$ for $\Theta$, every regret-minimizing policy $\pi_2$ for $\Theta$, and every $\theta\in\Theta$,
\begin{equation*}
\max\bigl\{W(\theta,\pi_1),W(\theta,\pi_2)\bigr\}\leq \max_{\pi\in\Pi}W(\theta,\pi)\leq \max_{\pi\in\Pi}W(\theta_0,\pi)+\epsilon~.
\end{equation*}
\end{proposition}
\begin{proof}
We apply the upper-semicontinuity part of Berge's maximum theorem.\footnote{See Theorem 2 in Section 6.3 (Maximum Theorem) of  \citet*{Berge1963}.}  If a function of two variables is upper semicontinuous and the feasible region of the second variable depends on the first through a compact-valued, upper hemicontinuous correspondence, then taking the supremum over the second variable yields an upper semicontinuous function of the first.
By Lemma~\ref{lemma:joint_upper_W}, $W$ is upper semicontinuous. The feasible region is the fixed compact set $\Pi$ independent of the type $\theta$. Therefore, $\theta\mapsto \max_{\pi\in \Pi} W(\theta,\pi)$ is upper semicontinuous. Note that for any collection of receiver types, the expected utility generated by either a max-min or a regret-minimizing policy is pointwise bounded above by the optimal expected utility for the corresponding type.
\end{proof}

An alternative, more constructive proof is provided in Appendix~\ref{app:alter_detailed_proofs}. That proof combines the geometric properties of  the best-reply regions with a reusable technical device. For each action $a$, the set $B(\theta,a)$ is defined by a system of non-strict linear inequalities. 
Hence, as a correspondence in $\theta$, it is upper hemicontinuous (Lemma~\ref{lemma:upper_sets_appendix}).\footnote{ We use Painlev\'e--Kuratowski set limits and convergence, detailed in Definition~\ref{def:set_convergence} in Appendix~\ref{app:mode_proofs}. Throughout, we use the terms  \emph{upper} and \emph{lower semicontinuous} for real-valued functions and \emph{upper} and \emph{lower hemicontinuous} for set-valued correspondences.}
Consequently, when $\theta$ is sufficiently close to $\theta_0$, every posterior $\mu$ has a nearby posterior $\mu'$ such that $V(\theta_0,\mu')$ is approximately as high as $V(\theta,\mu)$.
The technical device rescales the weights assigned to the perturbed posteriors and adds a counterweight to
restore Bayes plausibility. It then produces a signal policy at $\theta_0$ with nearly the same expected utility as the original policy at $\theta$. The constructed policy need not be optimal at $\theta_0$. This gives the desired comparison and completes the proof.
Similar constructions appear in the related literature \citep[see, e.g.,][]{Kamenica2011BayesianPersuasion,arieli2023optimal,lipnowski2024perfect,lin2025generalized}.
The device is stated in Appendix~\ref{app:alt_proofs_geo}.

\subsection{Sufficient Condition for Continuity and Robustness}
\label{sec:sufficient}

In light of Proposition~\ref{lemma:set_upper_continuous}, it is enough to show that there exists a signal policy that guarantees almost the same expected utility as the optimal policy at $\theta_0$, regardless of the realized type, in order to establish continuity and robustness. 
Formally:

\begin{lemma}
\label{lemma:existence_for_all_type}
    Let $G(\theta_0)$ be a BP model.  
    If for every $\epsilon>0$, there exists $\delta>0$ and a signal policy $\pi^*$ such that $ W(\theta,\pi^*)\geq \max_{\pi\in\Pi} W(\theta_0,\pi)-\epsilon$ for every  $\theta \in \operatorname{Ball}_{\delta}(\theta_0)$, then $G(\theta_0)$ is continuous, MM-robust, and RM-robust.
\end{lemma}

The proof is provided in Appendix~\ref{sec:appen_sufficient_proofs}. 
The existence of $\pi^*$ implies lower semicontinuity of the sender's optimal value at $\theta_0$. 
Together with upper semicontinuity, it also implies that $\pi^*$ is almost optimal for every nearby type $\theta$. 
For any collection of types around $\theta_0$, the max-min policy and the regret-minimizing policy perform
at least as well as $\pi^*$ according to their respective criteria.

To construct the signal policy in Lemma~\ref{lemma:existence_for_all_type}, it suffices to require the existence of an optimal policy that can be perturbed.
Intuitively, for each supported belief of such an optimal policy and each corresponding induced action, the correspondence $\theta\mapsto B(\theta,a)$ should be lower hemicontinuous. 
Lower hemicontinuity of the corresponding system of linear inequalities requires both that $B(\theta,a)$ has dimension $|\Omega|-1$ as the simplex, and that no degenerate inequalities of the form $0\times\mu\leq0$ define $B(\theta,a)$ (duplicate actions for the receiver).\footnote{For any nonempty set $B$ and any $x_0 \in B$, $\dim(B):=\dim (L)$ where $L:=\text{span}\{x-x_0:x\in B\}$; if $B=\emptyset$, set $\dim(B)=-1$. }
For each type $\theta$, denote the group of duplicate actions with the same receiver payoff as $a$ by
\begin{equation}
    R(\theta,a):=\{b\in A: u_{\theta}(a,\omega)= u_{\theta}(b,\omega), \forall \omega\in \Omega\}~.
\end{equation}
However, duplicate actions for the receiver may be utility-equivalent for the sender, either globally or only at a particular posterior belief.
This motivates the following definition.
\begin{definition}
\label{def:critical_mu}
 Let $G(\theta_0)$ be a BP model.  For each $\mu\in \Delta(\Omega)$, we say that $\mu$ is \emph{not critical at $\theta_0$} if there exists an induced action $a$ such that 
 \begin{enumerate}
     \item  $\dim(B(\theta_0, a))=|\Omega|-1$. 
     \item For each $b\in R(\theta_0,a)$, $\sum_{\omega\in\Omega} v(b,\omega)\mu(\omega)=\sum_{\omega\in\Omega} v(a,\omega)\mu(\omega)$.\footnote{If $R(\theta_0,a)=\{a\}$, this condition is automatically satisfied.}
 \end{enumerate}
 \label{def:property_B_a}
\end{definition}

\begin{example}
\label{ex:main_text_eample}
   Consider a BP model whose receiver and sender utilities are shown in Figure~\ref{table:uvalue_example_sec1} and where the sender's corresponding indirect utility function is depicted in Figure~\ref{fig:vvalue_example_sec1}. 
The model has five actions, including a pair of duplicate actions, $a_4$ and $a_5$, and an action, $a_2$, whose best-reply region is lower-dimensional.
\begin{figure}[!h]
  \centering
  \begin{subfigure}[b]{0.5\linewidth}
    \begin{tabular}{c|c|c|c|c|c}
\hline
$u_{\theta}(a,\omega)$ & \textbf{$a_1$} & { \textbf{$a_2$}}& \textbf{$a_3$}& \textbf{$a_4$}& \textbf{$a_5$}\\ \hline
$\omega=0$ & 0.7 & $ 0.6 $& 0.5 &-0.1 &-0.1   \\ \hline
${\omega}=1$ & -0.1 & 0.2 & $0.5$ &0.7 &0.7   \\ \hline
\multicolumn{1}{p{2cm}}{}\\ \hline 
$v(a,\omega)$ & {$a_1$} & { \textbf{$a_2$}}& \textbf{$a_3$}& \textbf{$a_4$}& \textbf{$a_5$}\\ \hline
$\omega=0$ & 0.15 & 0.18 & 0  & 0.35& 0.275\\ \hline
$\omega=1$ & 0.15 & 0.18 & 0  &0.35  & 0.375\\ \hline
\end{tabular}
    \caption{receiver and sender utilities}
    \label{table:uvalue_example_sec1}
  \end{subfigure}\hfill
  \begin{subfigure}[b]{0.5\textwidth}
      \tikzset{every picture/.style={inner sep=0, outer sep=0}}
    \begin{tikzpicture}[ scale = 0.3]

\draw[white] (-7,0)--(-0.5,0);

\node (a11) at (2,11)[]{$a_1$};
\node (a11) at (4,11)[]{$a_2$};
\node (a11) at (8,11)[]{$a_3$};
\node (a11) at (14,11)[]{$a_4,a_5$};

\draw[->,black] (0,0)--(0,11);
\draw[->,black] (0,0)--(16,0);

\draw[red,very thick] (0,3)--(4,3);
\draw[blue, very thick] (4,0)--(4,3.6);
\draw[blue, dashed,very thick] (12,7)--(16,7);
\draw[red,very thick] (12,7)--(16,7.5);

\draw[<->, gray] (0,9)--(4,9);
\draw[<->, gray] (4,9.5)--(12,9.5);
\draw[<->, gray] (12,10)--(16,10);

\node (zero) at (0,0) [label={[below=4pt]$0$}]{};
\node (zero) at (16,0) [label={[below=4pt]$1$}]{};
\node (zero) at (4,0) [label={[below=4pt]$0.25$}]{};
\node (zero) at (12,0) [label={[below=4pt]$0.75$}]{};
\node (zero) at (8,0) [label={[below=4pt]$\mu_0$}]{};

\draw[black, dotted, very thick] (0,3)--(12,7);
   
    \end{tikzpicture}
    \caption{sender indirect utility $V(\theta,\mu)$}
    \label{fig:vvalue_example_sec1}
    \end{subfigure}
    \caption{Example 2.}
  \label{fig:example_sec1}
\end{figure}

Note that the sender has an optimal policy supported on two non-critical posteriors, $0$ and $0.75$.
Consequently, as formalized later in this section, this persuasion model is continuous and robust under both behavioral criteria.%
\footnote{Anecdotally, we mention that the BP model in this example does not satisfy the sufficient condition (or assumptions) introduced in \citet{lipnowski2024perfect,gan2025robust, lin2025generalized}.}
\end{example}

We use the following observation, whose proof is deferred to Appendix~\ref{sec:appen_sufficient_proofs}. 
It shows that there exists a measurable perturbation rule that works for all non-critical posteriors at $\theta_0$ and all nearby types. Under this rule, every non-critical posterior is perturbed slightly, while its benchmark indirect utility is preserved up to a small loss for every realized type sufficiently close to $\theta_0$.

\begin{lemma}
\label{lemma:intersect_small_set}
Let $G(\theta_0)$ be a BP model. For every $\gamma>0$, there exist $\delta>0$ and a Borel measurable mapping $T:\Delta(\Omega)\to\Delta(\Omega)$ such that for every non-critical $\mu$ at $\theta_0$ and every type $\theta\in \operatorname{Ball}_{\delta}(\theta_0)$, $\|\mu-T(\mu)\|_2\leq \gamma$ and $V(\theta,T(\mu))\geq V(\theta_0,\mu)-\sqrt{|\Omega|}\gamma$.
\end{lemma}

If an action $a$ has a full-dimensional best-reply region, then, together with all its duplicate actions $R(\theta_0,a)$, the correspondence $\theta\mapsto \bigcup_{b\in R(\theta_0,a)} B(\theta,b)$ is lower hemicontinuous at $\theta_0$. We also show in the proof of Lemma~\ref{lemma:intersect_small_set} that there exists $\theta'\in \operatorname{Ball}_{\delta}(\theta_0)$ such that, for any type $\theta\in \operatorname{Ball}_{\delta}(\theta_0)$, the union $\bigcup_{b\in R(\theta_0,a)} B(\theta',b) \subseteq \bigcup_{b\in R(\theta_0,a)} B(\theta,b)$.
Since  $\min_{b\in R(\theta_0,a)} \sum_{\omega\in\Omega}v(b,\omega)\mu(\omega)$ defines a continuous function of the belief $\mu$, moving a non-critical belief $\mu$ to a sufficiently close belief $\mu'$ contained in $\bigcup_{b\in R(\theta_0,a)}B(\theta',b)$ ensures that $V(\theta,\mu')$ is almost as large as $V(\theta_0,\mu)$, regardless of the realized type $\theta$.
The sufficient condition is then the following.
\begin{proposition}[Sufficient Condition for Continuity and Robustness]
\label{lemma:sufficient_puri_new}
\label{prop:suffi_all}
Let $G(\theta_0)$ be a BP model.  
If there exists an optimal signal policy $\pi$ such that every posterior belief in its support is non-critical at $\theta_0$, then $G(\theta_0)$ is continuous, MM-robust, and RM-robust.
\end{proposition}

The proof is provided in Appendix~\ref{sec:appen_sufficient_proofs}. 
By Lemma~\ref{lemma:intersect_small_set}, there exists a mapping that moves the support of an optimal policy to beliefs whose utility is guaranteed regardless of the exact type. 
The technical device in Appendix~\ref{app:alt_proofs_geo} then  constructs the signal policy required by Lemma~\ref{lemma:existence_for_all_type}, which completes the proof.

\subsection{Necessary Condition for Continuity and Robustness}
\label{sec:nece_example}

Proposition~\ref{lemma:set_upper_continuous} implies that continuity is necessary for robustness:
\begin{lemma}
    \label{prop:continuous_robust_nece}
If a BP model is MM-robust or RM-robust, then it is also continuous.  
\end{lemma}
The proof is provided in Appendix~\ref{app:nece_proofs}. 
It follows immediately that if a BP model is not continuous, then its discontinuity can only be a downward jump in the sender's optimal value; consequently, robustness fails.
We now show that the conditions in Proposition~\ref{prop:suffi_all} are also necessary.
\begin{proposition}[Necessary Condition for Continuity and Robustness]
\label{prop:necessary_all}
 Let $G(\theta_0)$ be a BP model.
 If the conditions in Proposition \ref{lemma:sufficient_puri_new} fail, then there exists a positive constant $C^*$ and  a sequence of receiver types $\{\theta_n\}_{n\in \mathbb{N}^+}$ converging to $\theta_0$ such that $\displaystyle \limsup_{n\rightarrow\infty} \left[ \max_{\pi\in \Pi}W(\theta_n,\pi)\right] \leq\max_{\pi\in \Pi}W(\theta_0,\pi)-C^*$.
Hence, $G(\theta_0)$ is neither continuous nor robust under either behavioral  criterion.
\end{proposition}

The construction of this adverse sequence relies on a geometric observation. For each action $a$, if $B(\theta,a)$ is nonempty and lower-dimensional, 
then there exists another action $b$ with a full-dimensional best-reply region such that $B(\theta,a)\subseteq B(\theta,b)$. Consequently, the full-dimensional best-reply regions cover the entire simplex. This observation is stated as Lemma~\ref{lemma:lower_dimensional_set_subset} in Appendix~\ref{app:nece_proofs}.

We prove Proposition~\ref{prop:necessary_all} in two steps. First, we define a \emph{virtual type} and show in Lemma~\ref{lemma:nece_pseudo_gap} that, when the conditions in Proposition~\ref{prop:suffi_all} fail, the optimal value at the virtual type is lower than that at the original type by a positive constant.
Then, in Lemma~\ref{lemma:nece_sequence}, we approximate the virtual type by a sequence of types converging to $\theta_0$, obtaining an adverse sequence as required by Proposition~\ref{prop:necessary_all}.

The virtual type represents a \emph{virtual receiver} whose best-reply regions are constructed from those of the original type by two modifications: we remove all lower-dimensional regions and, at each posterior, retain only the sender-worst actions within each duplicate class. Using only the resulting collection of regions, we define the sender's indirect utility function---and hence her expected utility function---at the virtual type without specifying receiver utilities. Formally,

\begin{definition}
    \label{def:pseudo_type}
 Let $G(\theta)$ be a BP model.  The virtual type $\sigma(\theta)$ is defined as follows. If $B(\theta,a)=\emptyset$ or $\dim(B(\theta,a))<|\Omega|-1$, then $B(\sigma(\theta),a)=\emptyset$; otherwise,
  \begin{equation*}
      B(\sigma(\theta),a)=\{\mu\in B(\theta,a):\sum_{\omega\in\Omega} v(a,\omega)\mu(\omega)\leq \sum_{\omega\in\Omega} v(b,\omega)\mu(\omega),\text{ for every } b\in R(\theta,a)  \}~.
  \end{equation*}
The virtual best-reply regions cover every posterior belief.\footnote{Note that $\bigcup_{b\in R(\theta,a)}   B(\sigma(\theta),b)= B(\theta,a)$ for every action $a$ such that $\dim(B(\theta,a))=|\Omega|-1$.}
For every posterior belief $\mu$, the receiver's best replies are $  A^*(\sigma(\theta),\mu):=\{a\in A: \mu\in B(\sigma(\theta),a)\}$.
The indirect utility function is $V(\sigma(\theta),\mu):=\max_{a\in  A^*(\sigma(\theta),\mu)} \left[ \sum_{\omega\in \Omega} v(a,\omega)\mu(\omega)\right]$.
For every signal policy $\pi$, the expected utility is $W(\sigma(\theta),\pi):=\int_{\Delta(\Omega)}V(\sigma(\theta),\mu) \pi(d\mu)$.
\end{definition}
\begin{remark}
   The virtual type need not be generated by a receiver utility function.
\end{remark}
To gain intuition, consider a BP model of type $\theta_0$, whose best-reply regions are shown on the left of Figure~\ref{fig:type_violation} (satisfying Lemma~\ref{lemma:lower_dimensional_set_subset}).
There is a pair of duplicate actions, $a_1$ and $a_5$, and a lower-dimensional best-reply region corresponding to action $a_6$. 
The virtual type $\sigma(\theta_0)$ is shown on the right of Figure~\ref{fig:type_violation}. 

\begin{figure}[h!]
\centering
\begin{subfigure}  [b]{0.45\textwidth}
\centering
\begin{tikzpicture}[baseline = 3cm, xscale=0.19, yscale = 0.19]
        \draw[] (0,0) -- (20,0);
         \draw[] (0,0) -- (10,17.3);
           \draw[]  (10,17.3) -- (20,0);
             \draw[black, dashed, thick]  (9,7) -- (6,10.38);
    \draw[black, dashed, thick]  (9,7) -- (2,0);
     \draw[black, dashed, thick]  (9,7) -- (17,5.19);
     \draw[red,thick]  (15,0) -- (19,1.7);
   
         \draw[<-,thick] (0.4,5) -- (4,5);
         \draw[<-,thick] (18,-2) -- (18,0.5);

         \draw[->,thick]  (19,1.7) -- (20,6);

            \node[] at (10.3,12) (9) {$a_3$};   
        \node[] at (-2,5) (9) {$a_2$};    
       \node[] at (12,3) (9) { $a_1$, $a_5$};      
       \node[] at (18,-3) (9) {$a_4$};  
       \node[] at (22,7) (9) {$a_6$};  
     
    \end{tikzpicture}
    \caption{original type $\theta_0$}
    \end{subfigure}
\begin{subfigure}  [b]{0.45\textwidth}
\centering
\begin{tikzpicture}[baseline = 3cm, xscale=0.19, yscale = 0.19]
        \draw[] (0,0) -- (20,0);
         \draw[] (0,0) -- (10,17.3);
           \draw[]  (10,17.3) -- (20,0);
             \draw[black, dashed, thick]  (9,7) -- (6,10.38);

    \draw[black, dashed, thick]  (9,7) -- (2,0);
     \draw[black, dashed, thick]  (9,7) -- (17,5.19);
     \draw[black, dashed, thick]  (15,0) -- (19,1.7);
   \draw[blue, dashed, thick]  (11,0) -- (12,6);
         \draw[<-,thick] (0.4,5) -- (4,5);
         \draw[<-,thick] (18,-2) -- (18,0.5);
        
   \node[] at (10.3,12) (9) {$a_3$};  
        \node[] at (-2,5) (9) {$a_2$};
          \node[] at (8.5,3) (9) {$a_1$};      
       \node[] at (14.5,3) (9) {$a_5$};  
 \node[] at (18,-3) (9) {$a_4$};  
       
    \end{tikzpicture}
    \caption{virtual type $\sigma(\theta_0)$}
    \end{subfigure}
     \caption{Best-reply regions.}
     \label{fig:type_violation}
\end{figure}
The construction deliberately retains the sender-worst action in each duplicate class.
In particular, for any belief $\mu\in B(\theta_0,a_1)$, if the sender strictly prefers action $a_1$ to action $a_5$ at $\mu$, then $\mu\in B(\sigma(\theta_0),a_5)$; if she strictly prefers action $a_5$ to action $a_1$ at $\mu$, then $\mu\in B(\sigma(\theta_0),a_1)$; otherwise, $\mu\in B(\sigma(\theta_0),a_1)$ and $\mu\in B(\sigma(\theta_0),a_5)$.
Hence, any policy that induces the receiver to choose the sender-preferred action among $a_1$ and $a_5$ at $\theta_0$ cannot induce that action at $\sigma(\theta_0)$. In addition, the lower-dimensional set is removed, $B(\sigma(\theta_0),a_6)=\emptyset$.
Again, any policy that induces the receiver to choose $a_6$ at $\theta_0$ cannot induce that action at $\sigma(\theta_0)$.\footnote{We revisit the virtual type in Section~\ref{sec:tie_eq}, where we show that it can also be defined strategically.}

These two modifications ensure that the expected utility generated by any policy at $\sigma(\theta)$ is no greater than that generated by the same policy at $\theta$. Taking the supremum over policies on both sides of this pointwise inequality yields $\max_{\pi\in\Pi}W(\sigma(\theta),\pi)\leq\max_{\pi\in\Pi}W(\theta,\pi)$ (the upper-semicontinuity of $\mu\mapsto V(\sigma(\theta),\mu)$ is stated as Lemma~\ref{lemma:v_pse_upper} in Appendix~\ref{app:nece_proofs}).
If the conditions in Proposition~\ref{lemma:sufficient_puri_new} fail, this weak inequality between the optimal values at $\sigma(\theta)$ and $\theta$ becomes strict:

\begin{lemma}
\label{lemma:nece_pseudo_gap}
 Let $G(\theta_0)$ be a BP model.
 If the conditions in Proposition \ref{lemma:sufficient_puri_new} fail, then there exists a positive constant $C^*$ such that the virtual type $\sigma(\theta_0)$ defined in Definition~\ref{def:pseudo_type} satisfies $ \max_{\pi\in \Pi}W(\sigma(\theta_0),\pi)\leq\max_{\pi\in \Pi}W(\theta_0,\pi)-C^*$.
\end{lemma}

The proof is provided in Appendix~\ref{app:nece_proofs}. 
The intuition is that failure of the conditions in Proposition~\ref{lemma:sufficient_puri_new} implies that every optimal policy at $\theta_0$ has a critical belief in its support at which the receiver has an available action that is inferior from the sender's perspective (see Lemma~\ref{lemma:inf_choice} in Appendix~\ref{app:nece_proofs}). 
When moving from $\theta_0$ to $\sigma(\theta_0)$, the construction retains only such sender-inferior choices at critical beliefs (Lemma~\ref{lemma:inf_choice_2}), as illustrated in the preceding example.

Additionally, we use the fact that there is always an optimal policy supported on extreme points of best-reply regions whose associated weights are uniquely determined by the supporting extreme points (Lemma~\ref{lemma:representation_basic} in Appendix~\ref{app:nece_proofs}); we call such policies basic policies (Definition~\ref{def:basic_policies}).\footnote{For each $a$ and each $\theta$, both $B(\theta,a)$ and $B(\sigma(\theta),a)$ are generated by finite linear inequalities, and so have finite extreme points.}
So it suffices to consider finitely many basic policies at $\sigma(\theta_0)$. 
A basic policy at $\sigma(\theta_0)$ is either non-optimal at $\theta_0$, or it is optimal at $\theta_0$ but, at $\sigma(\theta_0)$, the receiver is forced to choose an action that is strictly inferior for the sender at some belief in the policy's support.
Hence, each such policy yields a strictly lower payoff at $\sigma(\theta_0)$ than the optimal value at $\theta_0$.

The following lemma, which shows that the virtual type $\sigma(\theta_0)$ can be approximated by a sequence that converges to $\theta_0$, completes the proof of Proposition~\ref{prop:necessary_all}.
 
\begin{lemma}
\label{lemma:nece_sequence}
For any BP model $G(\theta_0)$, consider the following sequence of types: for every $n$ and every pair $(a,\omega)$,
 \begin{equation*}
     u_{\theta_n}(a,\omega):=
     \begin{cases}
         u_{\theta_0}(a,\omega) + \tfrac{2}{n} -\tfrac{1}{n}v(a,\omega) & \text{if } B(\theta_0,a)\neq\emptyset \,,\,\dim(B(\theta_0,a))=|\Omega|-1
         \\
         u_{\theta_0}(a,\omega) - \tfrac{1}{n} & \text{otherwise} 
     \end{cases}~.
 \end{equation*}
Then $\displaystyle\limsup_{n\rightarrow \infty} \left[ \max_{\pi\in \Pi}W(\theta_n,\pi)\right]\leq \max_{\pi\in \Pi}W(\sigma(\theta_0),\pi)$.
\end{lemma}
The formal proof is deferred to Appendix~\ref{app:nece_proofs}. The argument is similar to the upper semicontinuity results (Proposition~\ref{lemma:set_upper_continuous}). The distinction is that the argument here applies to a particular sequence, and the limiting upper bound is evaluated at the virtual type $\sigma(\theta_0)$ rather than at the limiting real type $\theta_0$.
The key step is to show that for every posterior belief $\mu$ and every sequence $\{\mu_n\}$ with $\mu_n\to\mu$, $\displaystyle\limsup_{n\rightarrow \infty} V(\theta_n,\mu_n)\leq V(\sigma(\theta_0),\mu)$. 
Thus, along the specified sequence of types, the indirect utilities are asymptotically bounded above by the indirect utility of the virtual type. The remainder of the proof follows by adapting the techniques used to prove Lemma~\ref{lemma:joint_upper_W}, concerning the upper semicontinuity of $W$, and Proposition~\ref{lemma:set_upper_continuous}.

Lemma~\ref{lemma:nece_sequence} also has a geometric interpretation. Along the sequence constructed above, the best-reply region of each action is almost contained in its counterpart at $\sigma(\theta_0)$. Formally, for every $a\in A$, $\displaystyle\limsup_{n\to\infty}B(\theta_n,a) \subseteq B(\sigma(\theta_0),a)$.\footnote{The set convergence notion is stated in Definition~\ref{def:set_convergence} in Appendix~\ref{app:mode_proofs}.}
In the example in Figure~\ref{fig:type_violation}, since $B(\theta_0,a_6)\subseteq B(\theta_0,a_4)$, slightly decreasing the receiver's utility from $a_6$, while increasing his utility from actions with full-dimensional best-reply regions, makes another action strictly preferable to $a_6$ at every belief (formally, Lemma~\ref{lemma:nece_remove_clear}). Hence, $B(\theta_n,a_6)=B(\sigma(\theta_0),a_6)=\emptyset$. 
For $a_1$ and $a_5$, $u_{\theta_0}(a_1,\omega)=u_{\theta_0}(a_5,\omega)$ for all states at $\theta_0$. The perturbation therefore transforms their best-reply comparison into
\begin{equation*}
    \sum_{\omega\in \Omega} u_{\theta_n}(a_1,\omega) \mu(\omega) 
    \geq 
    \sum_{\omega\in \Omega} u_{\theta_n}(a_5,\omega) \mu(\omega) 
    \Longleftrightarrow  
    \sum_{\omega\in \Omega} v(a_1,\omega) \mu(\omega) 
    \leq 
    \sum_{\omega\in \Omega} v(a_5,\omega) \mu(\omega)~.
\end{equation*}
This is precisely the constraint used to define the virtual type.  A formal version of this geometric interpretation (Lemma~\ref{lemma:nece_corre_alter}) is provided in Appendix~\ref{app:alt_proofs_geo}.

Lemmas~\ref{lemma:nece_pseudo_gap}~and~\ref{lemma:nece_sequence} complete the proof of Proposition~\ref{prop:necessary_all}.
We conclude that the sufficient conditions in Proposition~\ref{lemma:sufficient_puri_new} are also necessary for continuity and robustness. Thus, these notions are equivalent.
We summarize the necessary and sufficient conditions as follows. 
\begin{proposition}[Necessary and Sufficient Conditions for Continuity and Robustness]
\label{pro:nece_su_summary}
    For any BP model $G(\theta_0)$, continuity, MM-robustness, and RM-robustness hold if and only if there exists an optimal signal policy $\pi$ such that every posterior belief in its support is non-critical at $\theta_0$.
\end{proposition}

\section{Tie-Breaking, Genericity, and Non-Robustness}
\label{sec:tie_eq}
We now establish the connection between robustness to receiver preferences and robustness to tie-breaking rules and discuss the genericity of robustness.
\begin{definition}
\label{defi:robust_tie_breaking}
A BP model $G(\theta)$ is robust to tie-breaking rules if
\begin{equation*}
    \sup_{\pi\in\Pi} \left [ \int_{\Delta(\Omega)}  \min_{a\in A^*(\theta,\mu)}  \left(\sum_{\omega\in \Omega} \mu(\omega) v(a,\omega) \right)  \pi(d\mu) \right ]= \max_{\pi\in\Pi} W(\theta,\pi)~.
\end{equation*}
\end{definition}

When the receiver has multiple best replies, his choice may be determined by a preference that is not captured by his utility function and may vary across decision instances. 
Recall that the assumption that the receiver's utility function is common knowledge simplifies an environment in which the sender is uncertain about the receiver's preferences. 
Analogously, a prespecified tie-breaking rule is a simplification of this underlying selection process. 
We show that robustness to tie-breaking rules---informally, robustness to such secondary preferences---is equivalent to robustness to preferences represented by his utility function.

\begin{Theorem}
 \label{thy:equive_tie_breaking}
For any BP model $G(\theta_0)$, the statement that $G(\theta_0)$ is robust to tie-breaking rules is equivalent to those in Theorem~\ref{theorem:bp_equivalent}.
\end{Theorem}

The equivalence is established by the fact that all notions have the same necessary and sufficient conditions shown in Proposition~\ref{pro:nece_su_summary}.

\begin{proposition}
\label{prop:nece_suf_tie_breaking}
A BP model $G(\theta_0)$ is robust to tie-breaking if and only if there exists an optimal signal policy $\pi$ such that all posterior beliefs in its support are non-critical at $\theta_0$.
\end{proposition}

Unlike the potentially unique best response (PUBR) condition of \citet{lipnowski2024perfect}, which is sufficient for robustness to tie-breaking, our necessary and sufficient condition permits an optimal policy to induce duplicate receiver actions, even though such actions can never be uniquely optimal for the receiver. Recall Example~\ref{ex:main_text_eample}, in which an optimal policy is supported on the posterior beliefs $0$ and $0.75$, both of which are non-critical. At the latter posterior, the sender is indifferent between the receiver's duplicate actions, although this indifference does not hold at other posteriors. Nevertheless, the preceding proposition implies that the BP model is robust to tie-breaking.

We provide the formal proof in Appendix~\ref{app:tie_eq}. 
The sufficiency argument is analogous to the proof of Proposition~\ref{prop:suffi_all}. It constructs signal policies whose sender-worst expected utility approaches the benchmark optimum.

For the necessary direction, note that the virtual type constructed in Section~\ref{sec:nece_example} can be interpreted as a particular tie-breaking rule. Formally,
\begin{definition}
\label{def:special_tie}
For every BP model $G(\theta)$ and posterior belief $\mu$, among the receiver's best replies $A^*(\theta,\mu)$, the admissible action under $\Sigma$ satisfies:
 \begin{enumerate}
     \item $\dim(B(\theta,a))=|\Omega|-1$. 
     \item $\sum_{\omega\in \Omega} v(a,\omega)\mu(\omega)\leq \sum_{\omega\in \Omega} v(b,\omega)\mu(\omega)$ for every $b\in R(\theta,a)$.
     \item $\sum_{\omega\in \Omega} v(a,\omega)\mu(\omega)\geq \min_{c\in R(\theta,b)} \sum_{\omega\in\Omega}v(c,\omega)\mu(\omega)$ for every $b\in A^*(\theta,\mu)$ such that $\dim(B(\theta,b))=|\Omega|-1$.
 \end{enumerate}
\footnote{Lemma~\ref{lemma:lower_dimensional_set_subset} supplies a full-dimensional best reply; finiteness then guarantees a sender-worst duplicate in each class and a sender-best surviving class. If multiple actions satisfy these three properties, fix an ordering of $A$ and select the action with the smallest index.}
Let $V_{\Sigma}:\mathbb{R}^{|\Omega|\times |A|}\times\Delta(\Omega)$ denote the sender's indirect utility under this tie-breaking rule.  
\end{definition}
Under $\Sigma$, the receiver follows a three-stage selection rule: he first discards best replies whose best-reply regions are lower-dimensional; then, within each exact-duplicate class, he retains the sender-worst actions; finally, among the surviving actions across different duplicate classes, he selects the best for the sender.\footnote{The selection rule is reminiscent of that studied in Stackelberg games by \citet{von2010leadership}.}

\begin{lemma}
    \label{lemma:eq_tie_pseudo}
    For any BP model $G(\theta)$ and any posterior belief $\mu$, the virtual type (Definition~\ref{def:pseudo_type}) and the tie-breaking rule in Definition~\ref{def:special_tie} satisfy $V(\sigma(\theta),\mu)=V_{\Sigma}(\theta,\mu)$.
\end{lemma} 
\begin{proof}
$A^*(\sigma(\theta),\mu)$ consists of the full-dimensional best replies at $\theta$ that are worst for the sender within their respective duplicate classes. Then the sender-optimal action within $A^*(\sigma(\theta),\mu)$ defines $V(\sigma(\theta),\mu)$. This is precisely $\Sigma$.
\end{proof}

The necessity direction follows directly by combining the preceding lemma, Lemma~\ref{lemma:nece_pseudo_gap} and the pointwise inequality $\min_{a\in A^*(\theta,\mu)}\sum_{\omega\in\Omega}\mu(\omega)v(a,\omega)\leq V_{\Sigma}(\theta,\mu)$.

Under the sender-worst tie-breaking, the function $\mu\mapsto\min_{a\in A^*(\theta,\mu)}\sum_{\omega\in\Omega}\mu(\omega)v(a,\omega)$ need not be upper semicontinuous. Consequently, the existence of a signal policy attaining the supremum under this tie-breaking rule is not guaranteed. By contrast, $V_{\Sigma}(\theta,\mu)$ is upper semicontinuous in $\mu$ (Lemma~\ref{lemma:v_pse_upper} in Appendix~\ref{app:nece_proofs}). We further show that $\Sigma$ is the \emph{tight} tie-breaking rule to test robustness.

\begin{lemma}
\label{lemma:tight_tie_breaking}
For every BP model $G(\theta)$,  the tie-breaking rule $\Sigma$ satisfies
\begin{equation*}
    \displaystyle \max_{\pi\in \Pi} \left[ \int_{\Delta(\Omega)} V_{\Sigma}(\theta,\mu)\pi(d\mu)\right]=\sup_{\pi\in\Pi} \left [ \int_{\Delta(\Omega)}  \min_{a\in A^*(\theta,\mu)}  \left(\sum_{\omega\in \Omega} \mu(\omega) v(a,\omega) \right)  \pi(d\mu) \right ]~.
\end{equation*}
\end{lemma}
The proof is provided in Appendix~\ref{app:tie_eq}. 
To gain intuition, every action selected under $\Sigma$ has a full-dimensional best-reply region. The corresponding posterior belief can therefore be perturbed into the interior of that region, where the same sender payoff is approximately guaranteed under arbitrary tie-breaking. The construction is analogous to the proof of the sufficient direction in Proposition~\ref{prop:nece_suf_tie_breaking}.

The preceding lemma and the equivalence established in Theorems~\ref{theorem:bp_equivalent}~and~\ref{thy:equive_tie_breaking} yield a new numerical necessary and sufficient condition for all notions. 

\begin{proposition}[Numerical Necessary and Sufficient Condition]
    \label{prop:num_nece_suff}
    For any BP model $G(\theta_0)$, continuity, MM-robustness, RM-robustness, and robustness to tie-breaking hold if and only if $\displaystyle \max_{\pi\in \Pi} \left[ \int_{\Delta(\Omega)} V_{\Sigma}(\theta_0,\mu)\pi(d\mu)\right]=\max_{\pi\in\Pi} W(\theta_0,\pi)$.
\end{proposition}

Finally, we show that continuity and robustness hold generically.
We treat utility values $\{u_{\theta}(a,\omega)\}_{a\in A, \omega\in \Omega}$ as a set of random variables.

\begin{definition}
\label{defi}
Let $U$ be a property of BP models. We say that $U$ is \emph{generic} if the set of receiver utility profiles for which $U$ fails has Lebesgue measure zero in $\mathbb{R}^{|\Omega|\times|A|}$. Equivalently, $U$ holds for Lebesgue-almost every receiver utility profile.
\end{definition}

\citet{lipnowski2024perfect} show that, generically, the Bayesian persuasion model is robust to tie-breaking rules (see their Proposition 2). 
Together with the equivalence results in Theorems~\ref{theorem:bp_equivalent}~and~\ref{thy:equive_tie_breaking}, we obtain the following direct corollary.
\begin{corollary}
\label{thy:generic_robust}
\textbf{Generically}, the Bayesian persuasion model is continuous, MM-robust, and RM-robust. 
\end{corollary}

\subsection{Two Caveats}
\label{sec:final_discuss}
Although our main result is a strong positive result about the robustness of the Bayesian persuasion model, we would like to provide two comments so that this result is taken with a grain of salt.

First, the level of accuracy, $\epsilon$, depends on the level of ignorance, $\delta$. A stronger notion of robustness (a ``uniform'' notion) could require this relationship to be independent of the underlying BP model. Unfortunately, this does not hold (even for bounded receiver utilities): for any given level of accuracy $\epsilon$ one can find a sequence of BP models such that the required level of ignorance $\delta$ converges to zero. Example~\ref{example:appendix_noUniform} in Appendix~\ref{app:tie_eq} illustrates this impossibility.\footnote{In particular, take any non-robust BP model and a disappearing sequence of neighborhoods. Every such neighborhood of that model,  no matter how small, contains a robust BP model. 
Choosing one robust model from each neighborhood yields the required sequence.}

Second, \emph{genericity} is a mathematical property that hinges on the assumption that instances of BP models are generated by a given randomization mechanism. Should economists take this underlying randomness for granted? We provide two classes of BP models that could naturally occur and that are not robust. 

{\bf Two-Layer Bureaucracy} It is often the case that in application systems there are two gatekeepers. The first one handles the mass of requests and only unresolved cases are forwarded to the second layer. 
An applicant is in one of two states:  competent or incompetent. The first gatekeeper processes a request and can accept, deny, or forward it to the second gatekeeper. Immediate acceptance is optimal for the applicant because further scrutiny by the second gatekeeper is costly (e.g., time-consuming). In many bureaucracies, institutional gatekeepers are deterred from making a mistake of type one---approving a request that should be denied---while consequences of the opposite mistake are brushed over. In such settings, the gatekeeper will only approve if he is certain that the applicant is competent. Note that this is exactly the story underlying Example~\ref{sec:warm_examples}.

{\bf A Preference for Compromise} In some binary-action settings, a third natural action may exist, an action which serves as a compromise between the initial two. For example, in an allocation problem where a good can be allocated to one of two individuals, a third option of flipping a fair coin may also be acceptable. A compromise action need not take the form of a lottery. For example, a venture-capital fund may be interested in backing an entrepreneur who is looking to raise a certain amount. The fund can either invest, in return for shares, or provide the required amount as a loan. A third natural option is to provide half the amount as an investment and half as a loan.

\begin{figure}[!h]
    \centering
      \begin{subfigure}[b]{0.35\textwidth}
      \tikzset{every picture/.style={inner sep=0, outer sep=0}}
    \begin{tikzpicture}[baseline = 3cm, scale = 0.27]

     \node (a11) at (0,0)[label={[left=1pt]\textcolor{black}{$a_2$}}]{};
\node (a21) at (0,5) [label={[left=1pt]\textcolor{red}{$a_3$}}]{};
\node (a31) at (0,10) [label={[left=1pt]\textcolor{black}{$a_1$}}]{};

\node (zero) at (0,0) [label={[below=4pt]$0$}]{};
\node (zero) at (16,0) [label={[below=4pt]$1$}]{};
\node (zero) at (4.8,0) [label={[below=4pt]$\mu_0$}]{};

  \draw[->,black]  (0,0) -- (0,11);
     \draw[->,black]  (0,0) -- (18,0);

\draw[thick,black]  (16,10)  --  (8,5);
\draw[dashed, thick,black]  (8,5)  --  (0,0);
\draw[thick,black]  (0,10)  --  (8,5);
\draw[dashed, thick,black]  (8,5)  --  (16,0);
\draw[thick,dashed,red]  (0,5)  --  (16,5);

   \filldraw[red] (8,5) circle (10pt);

\node[] at (6,-2) {receiver's belief};      
        
    \end{tikzpicture}
    \caption{receiver's best replies}
    \label{fig:inconsistent_case_discuss1}

    \end{subfigure}
    \hspace{3cm}
       \begin{subfigure}[b]{0.35\textwidth}
      \tikzset{every picture/.style={inner sep=0, outer sep=0}}
    \begin{tikzpicture}[baseline = 3cm, scale = 0.27]

\node (zero) at (0,0) [label={[below=4pt]$0$}]{};
\node (zero) at (16,0) [label={[below=4pt]$1$}]{};
\node (zero) at (4.8,0) [label={[below=4pt]$\mu_0$}]{};

  \draw[->,black]  (0,0) -- (0,11);
     \draw[->,black]  (0,0) -- (18,0);

     \draw[<->,black]  (0,8) -- (8,8);
      \draw[<->,black]  (16,8) -- (8,8);
      \node (zero) at (4,10) []{$a_1$};
\node (zero) at (8,10) []{\textcolor{red}{$a_3$}};
\node (zero) at (12,10) []{$a_2$};

  \draw[red, very thick] (8,0) -- (8,6);
       \filldraw[red] (8,6) circle (8pt);    
\node[] at (6,-2) {receiver's belief};      
        
    \end{tikzpicture}
    \caption{sender's indirect utility}
    \label{fig:inconsistent_case_discuss2}

    \end{subfigure}
    \caption{A preference for compromise to destroy robustness.}
    \label{fig:inconsistent_case_discuss}
\end{figure}

It stands to reason to assume that the level of evidence that makes the decision maker (the receiver) indifferent between the two initial actions would make her indifferent between those two and the compromise action. If, in addition, the compromise action is the one favored by the sender (e.g., she cares about fairness or she enjoys the best of both worlds in the financing setting), we have an instance of Bayesian persuasion depicted in Figure~\ref{fig:inconsistent_case_discuss}, which is non-robust.

\appendix
\renewcommand{\theHsection}{appendix.\Alph{section}}
\renewcommand{\theHsubsection}{appendix.\Alph{section}.\arabic{subsection}}

\renewcommand{\theequation}{A.\arabic{equation}}
\renewcommand{\theHequation}{A.\arabic{equation}}
\setcounter{equation}{0}

\section{Set Convergence and Proofs for Section~\ref{sec:model}}
\label{app:mode_proofs}

\begin{definition}[Painlev{\'e}--Kuratowski set limits and convergence]
\label{def:set_convergence}
Following \citet[Definition 4.1]{rockafellar1998variational}, let $\{C_n\}_{n\in\mathbb{N}^+}$ be a sequence of subsets of $\mathbb{R}^t$. 
\begin{equation*}
    \text{The outer limit is: } \displaystyle      \limsup_{n\to\infty}C_n:=\left\{x\in\mathbb{R}^t:\liminf_{n\to\infty} \left[ \inf_{y\in C_n}||x-y||_2\right]=0\right\}~.
\end{equation*}
\begin{equation*}
    \text{The inner limit is: }   \displaystyle     \liminf_{n\to\infty}C_n:=\left\{x\in\mathbb{R}^t:\limsup_{n\to\infty}\left[\inf_{y\in C_n}||x-y||_2\right]=0\right\}~.
\end{equation*}
A correspondence $C:\mathbb{R}^q\rightrightarrows\mathbb{R}^p$ is upper (respectively, lower) hemicontinuous at $x_0$ if $x_n\to x_0$ implies $\limsup_n C(x_n)\subseteq C(x_0)$ (respectively, $\liminf_n C(x_n)\supseteq C(x_0)$); see \citet[Definition 5.4]{rockafellar1998variational}.\footnote{Their terms are \emph{outer} and \emph{inner semicontinuity}; we use the standard economic terminology.}
Following \citet[Theorem 4.10 and Proposition 5.12]{rockafellar1998variational}, for a closed-valued  correspondence $C:\mathbb{R}^q\rightrightarrows  K$ where $K$ is a compact subset of $\mathbb{R}^p$ (the best-reply region and simplex in our setting), $C$ is upper hemicontinuous at $x_0$ if and only if for every sequence $\{x_n\}_{n\in\mathbb{N}^+}$ with $x_n\to x_0$ and every $\epsilon>0$, there exists $N$ such that for every $n\geq N$,  if $C(x_0)\neq \emptyset$, $ \sup_{z\in C(x_n)}\inf_{y\in C(x_0)}||z-y||_2 \leq \epsilon$, and if $C(x_0)=\emptyset$, $C(x_n)=\emptyset$; it is lower hemicontinuous at $x_0\in \mathbb{R}^q$ if and only if, for every such sequence and every $\epsilon>0$, there exists $N$ such that, for every $n\geq N$, if $C(x_0)\neq \emptyset$,  $\sup_{z\in C(x_0)}\inf_{y\in C(x_n)}||z-y||_2 \leq \epsilon$, and if $C(x_0)=\emptyset$, the condition holds automatically.\footnote{Let $C\neq \emptyset$, define  $\sup_{z\in C}\inf_{y\in \emptyset}||z-y||_2:=\infty$ and $\sup_{z\in \emptyset}\inf_{y\in C}||z-y||_2 :=0$.}
\end{definition}

\begin{proof}[Proof of Lemma~\ref{lemma:Pi_compact}]
Consider $F(\pi):=\int_{\Delta(\Omega)}\mu\pi(d\mu)$, which is continuous on $\Delta(\Delta(\Omega))$ under the weak topology.   The singleton $\{\mu_0\}$ is closed in $\Delta(\Omega)$, so $\Pi=F^{-1}(\{\mu_0\})$ is closed in $\Delta(\Delta(\Omega))$. Since $\Delta(\Delta(\Omega))$ is compact, $\Pi$ is compact. 
\end{proof}

\begin{proof}[Proof of Lemma~\ref{lemma:V_upper_semi}]
Fix $(\theta,\mu)$. If $A^*(\theta,\mu)=A$, it is trivial. Otherwise, every $a\in A^*(\theta,\mu)$  yields strictly greater receiver expected utility than every $b\in A\setminus A^*(\theta,\mu)$.
For each action $a$, $(\theta,\mu)\mapsto\sum_{\omega\in\Omega}u_{\theta}(a,\omega)\mu(\omega)$ is continuous. 
So for any sequence $\{(\theta_n,\mu_n)\}_{n\in\mathbb{N}^+ }$ converging to $(\theta,\mu)$, there exists $I>0$ such that for all $n\geq I$,  $A^*(\theta_n,\mu_n)\subseteq A^*(\theta,\mu)$. 
For each action $a$, $\mu\mapsto\sum_{\omega\in\Omega}v(a,\omega)\mu(\omega)$ is continuous. So 
\begin{equation*}
    \displaystyle\limsup_{n\rightarrow \infty}\left[\max_{a\in A^*(\theta_n,\mu_n)  } \sum_{\omega\in\Omega}v(a,\omega)\mu_n(\omega) \right]\leq \displaystyle\limsup_{n\rightarrow \infty}\left[\max_{a\in A^*(\theta,\mu)}  \sum_{\omega\in\Omega}v(a,\omega)\mu_n(\omega)  \right]\leq V(\theta,\mu)~.
\end{equation*}
\end{proof}

\begin{proof}[Proof of Lemma~\ref{lemma:joint_upper_W}]
We endow $\mathbb{R}^{|\Omega|\times|A|} \times \Delta(\Omega)$ with the product metric and endow $\Delta(\mathbb{R}^{|\Omega|\times|A|}  \times \Delta(\Omega))$ with the weak topology.
Take any sequence $(\theta_n,\pi_n)\to(\theta,\pi)$ in $\mathbb{R}^{|\Omega|\times|A|} \times\Pi$. 
Define probability measures $q_n$ and $q$ on $\mathbb{R}^{|\Omega|\times|A|} \times  \Delta(\Omega)$ by
\begin{equation*}
 q_n(d\theta',d\mu):=\operatorname{Dirac}_{\theta_n}(d\theta') \pi_n(d\mu)
\qquad\text{and}\qquad
q(d\theta',d\mu):=\operatorname{Dirac}_{\theta}(d\theta') \pi(d\mu)~.   
\end{equation*}
We first show that $q_n\rightarrow q$. For any  bounded Lipschitz function $f(\theta,\mu)$,
\begin{equation*}
   \int_{\mathbb{R}^{|\Omega|\times|A|} \times \Delta(\Omega)} f(\theta',\mu)q_n(d\theta',d\mu) =\int_{\Delta(\Omega)} [f(\theta_n,\mu)-f(\theta,\mu)]\pi_n(d\mu)+\int_{\Delta(\Omega)} f(\theta,\mu)\pi_n(d\mu) ~.
\end{equation*}
Since $f$ is Lipschitz, the first term converges to zero.
For the second term, by weak convergence,  $\int_{\Delta(\Omega)} f(\theta,\mu)\pi_n(d\mu) \rightarrow \int_{\Delta(\Omega)} f(\theta,\mu)\pi(d\mu)=\int_{\mathbb{R}^{|\Omega|\times|A|} \times \Delta(\Omega)} f(\theta',\mu)q(d\theta',d\mu)$. Thus, by the Portmanteau theorem---bounded Lipschitz functions characterize weak convergence on metric spaces \citep*{billingsley2013convergence}, $q_n\rightarrow q$. \footnote{The Portmanteau theorem provides equivalent definitions of weak convergence of measures.}
By Lemma~\ref{lemma:V_upper_semi}, $V(\theta,\mu)$ is upper semicontinuous and bounded.
By the Portmanteau theorem,
\begin{equation*}
 \begin{split}
    {}&\limsup_{n\rightarrow \infty}    \int_{\mathbb{R}^{|\Omega|\times|A|} \times \Delta(\Omega)} V(\theta',\mu)\,q_n(d\theta',d\mu)\leq \int_{\mathbb{R}^{|\Omega|\times|A|} \times \Delta(\Omega)} V(\theta',\mu)q(d\theta',d\mu)~.
 \end{split}
 \end{equation*}
By the definitions of $q_n$, $q$, and $W$, it follows that $\displaystyle\limsup_{n\rightarrow \infty} W(\theta_n,\pi_n) \leq W(\theta,\pi)$.
\end{proof}

\renewcommand{\theequation}{B.\arabic{equation}}
\renewcommand{\theHequation}{B.\arabic{equation}}
\setcounter{equation}{0}

\section{Proofs for Section \ref{sec:main_equivelant}}
\label{sec:appen_equ_proofs}

\subsection{Sufficient Condition Proofs}
\label{sec:appen_sufficient_proofs}

\begin{proof}[Proof of Lemma~\ref{lemma:existence_for_all_type}]
By Proposition~\ref{lemma:set_upper_continuous} and the assumption, for any $\eta>0$, there exist $\delta>0$ and a signal policy $\pi^*$ such that, for every $\theta\in\operatorname{Ball}_{\delta}(\theta_0)$, $\max_{\pi\in\Pi} W(\theta_0,\pi)+\eta\geq \max_{\pi\in\Pi} W(\theta,\pi)\geq     W(\theta,\pi^*)\geq \max_{\pi\in\Pi} W(\theta_0,\pi)-\eta$.
Thus $G(\theta_0)$ is continuous. For any nonempty $\Theta\subseteq \operatorname{Ball}_{\delta}(\theta_0)$ and corresponding max-min policy $\pi_{1}$, since $ \inf_{\theta\in\Theta} W(\theta,\pi_1)\geq  \inf_{\theta\in\Theta}  W(\theta,\pi^*)\geq \max_{\pi\in\Pi} W(\theta_0,\pi)-\eta$, 
\begin{equation*}
\forall\theta\in \Theta: \quad \max_{\pi\in\Pi} W(\theta_0,\pi)+\eta\geq \max_{\pi\in\Pi} W(\theta,\pi)  \geq  W(\theta,\pi_1)\geq \max_{\pi\in\Pi} W(\theta_0,\pi)-\eta~.
\end{equation*}
Let $\pi_2$ be a regret-minimizing policy for $\Theta$, then  
\begin{equation*}
\sup_{\theta\in\Theta} \left[\max_{\pi\in\Pi}W(\theta,\pi)- W(\theta,\pi_2)\right]\leq \sup_{\theta\in\Theta} \left[\max_{\pi\in\Pi}W(\theta,\pi)- W(\theta,\pi^*)\right]\leq 2\eta~.
\end{equation*}
Thus, for every type $\theta\in\Theta$,  $ W(\theta,\pi_2) \geq \max_{\pi\in\Pi} W(\theta,\pi)-2\eta$. So $\max_{\pi\in\Pi} W(\theta_0,\pi)+\eta\geq \max_{\pi\in\Pi} W(\theta,\pi)\geq    W(\theta,\pi_2) \geq \max_{\pi\in\Pi} W(\theta_0,\pi)-3\eta$.
For any $\epsilon>0$, taking $\eta=\epsilon/4$ proves robustness under both criteria.
\end{proof}

\begin{proof}[Proof of Lemma~\ref{lemma:intersect_small_set}]
 For every $a$ and $b\in R(\theta_0,a)$, $u_{\theta_0}(a,\omega)=u_{\theta_0}(b,\omega)$. So 
 \begin{equation*}
  \bigcup_{b\in R(\theta_0,a)} B(\theta_0,b)=\{\mu\in \Delta(\Omega): \sum_{\omega\in \Omega} u_{\theta_0}(a,\omega)\mu(\omega)\geq \sum_{\omega\in \Omega} u_{\theta_0}(c,\omega)\mu(\omega), \forall c\in A\setminus R(\theta_0,a)\}~.    
 \end{equation*}
For any $\delta>0$ and any action $a$, define $\theta_a\in \operatorname{Ball}_{\delta}(\theta_0)$ as follows. If $b\not\in R(\theta_0,a)$, $  u_{\theta_a}(b,\omega)=u_{\theta_0}(b,\omega)+\delta$; otherwise, $ u_{\theta_a}(b,\omega)= u_{\theta_0}(b,\omega)-\delta$.
For every $\theta\in\operatorname{Ball}_{\delta}(\theta_0)$, $\bigcup_{b\in R(\theta_0,a)}B(\theta_a,b) \subseteq\bigcup_{b\in R(\theta_0,a)}B(\theta,b)$, because
for all $c\in A\setminus R(\theta_0,a)$, $\theta\in \operatorname{Ball}_{\delta}(\theta_0)$ and $\mu\in \bigcup_{b\in R(\theta_0,a)} B(\theta_a,b)$, 
\begin{equation*}
\min_{b\in R(\theta_0,a)} \sum_{\omega\in\Omega} u_{\theta}(b,\omega)\mu(\omega)\geq \sum_{\omega\in\Omega} u_{\theta_a}(a,\omega)\mu(\omega)\geq \sum_{\omega\in\Omega} u_{\theta_a}(c,\omega) \mu(\omega)\geq \sum_{\omega\in\Omega} u_{\theta}(c,\omega)\mu(\omega)~.
\end{equation*}
Note that $R(\theta_0,a)\subseteq R(\theta_a,a)$, and therefore
 \begin{equation*}
  \bigcup_{b\in R(\theta_0,a)} B(\theta_a,b)=\{\mu\in \Delta(\Omega): \sum_{\omega\in \Omega} u_{\theta_a}(a,\omega)\mu(\omega)\geq \sum_{\omega\in \Omega} u_{\theta_a}(c,\omega)\mu(\omega), \forall c\in A\setminus R(\theta_0,a)\}.
 \end{equation*}
After excluding duplicate actions, none of the inequalities is of the form $0\times \mu\leq 0$.

We next show that if $\dim(B(\theta_0,a))=|\Omega|-1$ and if $A\setminus R(\theta_0,a)\neq \emptyset$, then there exists a posterior $\mu\in \bigcup_{b\in R(\theta_0,a)} B(\theta_0,b)$ that satisfies all the remaining inequalities strictly. 
Suppose, for contradiction, that there are some $c\in A\setminus R(\theta_0,a)$ such that $\sum_{\omega\in\Omega}u_{\theta_0}(a,\omega)\mu(\omega)= \sum_{\omega\in\Omega}u_{\theta_0}(c,\omega)\mu(\omega)$ for every $\mu\in B(\theta_0,a)$.
Then $B(\theta_0,a)$ would be contained in the hyperplane corresponding to the equality,
which has dimension at most $|\Omega|-2$ in $\Delta(\Omega)$, contradicting
$\dim B(\theta_0,a)=|\Omega|-1$. So for every $c\in A\setminus R(\theta_0,a)$, there exists a posterior belief $\mu_c\in B(\theta_0,a)$ such that $\sum_{\omega\in\Omega}u_{\theta_0}(a,\omega)\mu_c(\omega)> \sum_{\omega\in\Omega}u_{\theta_0}(c,\omega)\mu_c(\omega)$. Let $\tilde{\mu}$ be a convex combination of $\{\mu_c\}_{c\in A\setminus R(\theta_0,a)}$ with strictly positive weights. Then by convexity, $\tilde{\mu}\in B(\theta_0,a)$ and  $\sum_{\omega\in\Omega}u_{\theta_0}(a,\omega)\tilde{\mu}(\omega)> \sum_{\omega\in\Omega}u_{\theta_0}(c,\omega)\tilde{\mu}(\omega)$ for every $c\in A\setminus R(\theta_0,a)$.
Relative to the simplex (equivalently, in its affine hull), the correspondence $\delta\mapsto\bigcup_{b\in R(\theta_0,a)} B(\theta_a,b)$ satisfies the Slater condition at $\delta=0$.

Hence, along the special path defined by $\theta_a$, the correspondence $\delta\mapsto\bigcup_{b\in R(\theta_0,a)} B(\theta_a,b)$ is lower hemicontinuous at $\delta=0$ (if $A\setminus R(\theta_0,a)\neq \emptyset$, see
Theorem 1.6 (basic equivalence) in \citep{daniilidis2013lower}; otherwise, $\bigcup_{b\in R(\theta_0,a)} B(\theta,b)=\Delta(\Omega)$ and the result is immediate).\footnote{\citet*{daniilidis2013lower} use both \emph{lower semicontinuity} and \emph{lower hemicontinuity} for different properties. Their lower semicontinuity corresponds to what we call lower hemicontinuity. Theorem~1.6 in their paper, which we invoke here, concerns precisely this property.} Since for every $\delta>0$, $\bigcup_{b\in R(\theta_0,a)} B(\theta_a,b)\subseteq \bigcup_{b\in R(\theta_0,a)} B(\theta,b)$ for every $\theta\in\operatorname{Ball}_{\delta}(\theta_0)$, it follows that $\theta\mapsto\bigcup_{b\in R(\theta_0,a)} B(\theta,b)$ is lower hemicontinuous at $\theta_0$. So for any $\gamma>0$, since $A$ is finite, we can choose $\delta>0$ such that for every $a\in A$, if $\dim(B(\theta_0,a))=|\Omega|-1$,  then $  \max_{x\in \bigcup_{b\in R(\theta_0,a)} B(\theta_0,b)} \min_{y\in \bigcup_{b\in R(\theta_0,a)} B(\theta_a,b)} ||x-y||_2\leq \gamma$.

For every $\mu\in \Delta(\Omega)$, define the mapping $T$ as follows. If $\mu$ is critical at $\theta_0$, let $T(\mu):=\mu$; otherwise, choose by a fixed ordering a sender-optimal $a\in A^*(\theta_0,\mu)$ satisfying the two conditions in Definition~\ref{def:critical_mu}, let $ T(\mu):= \operatorname{proj}(\mu,  \bigcup_{b\in R(\theta_0,a)} B(\theta_a,b))$.\footnote{Here, $\operatorname{proj}(\mu,S)$ denotes a closest point to $\mu$ in the set $S$. Note that for every action $a$, $\bigcup_{b\in R(\theta_0,a)} B(\theta_a,b)$ is convex and compact. The projection onto this set is well defined and unique.} So $||T(\mu)-\mu||_2\leq \gamma$.
For the latter case,  for every $\theta\in \operatorname{Ball}_{\delta}(\theta_0)$, $R(\theta_0,a)\bigcap A^*(\theta,T(\mu))\neq \emptyset$.

This mapping is Borel measurable because its domain ($\Delta(\Omega)$) can be partitioned into finitely many Borel regions, determined by linear inequalities specifying whether each $\mu$ is critical and which action it induces. On each region, the mapping is either a continuous projection or the identity.
Note that $v$ is bounded by $1$. $T$ satisfies, for every non-critical $\mu$ at $\theta_0$ and for every $\theta\in \operatorname{Ball}_{\delta}(\theta_0)$ 
\begin{equation*}
\begin{split}
     V(\theta,T(\mu)) &\geq \min_{b\in R(\theta_0,a)} \left[ \sum_{\omega\in\Omega} v(b,\omega)T(\mu)(\omega)  \right]
     \\
     {}&\geq \min_{b\in R(\theta_0,a)} \left[ \sum_{\omega\in\Omega} v(b,\omega)\mu(\omega) \right] - \sum_{\omega\in\Omega} \left|\mu(\omega)-T(\mu)(\omega)\right| \geq V(\theta_0,\mu) -\sqrt{|\Omega|}\gamma~.
\end{split}
\end{equation*}
\end{proof}

\begin{proof}[Proof of Proposition~\ref{prop:suffi_all}]
For any $\epsilon>0$, let $\gamma:={\epsilon}/({\sqrt{|\Omega|}+{1}/{L_{\mu_0}}})$.\footnote{$L_{\mu_0}:=\min_{\omega\in\Omega}\mu_0(\omega)$, see Lemma~\ref{lemma:far_point_around_u0}.} 
When the conditions in Proposition~\ref{prop:suffi_all} are satisfied, there exists an optimal signal policy $\pi$ such that all the posterior beliefs in its support are not critical. 
By Lemma~\ref{lemma:intersect_small_set}, there exists $\delta>0$ and a mapping $T:\Delta(\Omega)\to\Delta(\Omega)$ such that for any $\mu\in \operatorname{supp}(\pi)$ and for any
$\theta\in \operatorname{Ball}_{\delta}(\theta_0)$, $\|\mu-T(\mu)\|_2\leq  \gamma $ and $ V(\theta,T(\mu))\geq V(\theta_0,\mu)-\sqrt{|\Omega|} \gamma$.
By Corollary~\ref{coro:coro_adjust} in Appendix~\ref{app:alt_proofs_geo}, implementing $\pi_T$ guarantees that for any
$\theta\in \operatorname{Ball}_{\delta}(\theta_0)$,
\begin{equation*}
    W(\theta,\pi_T)\geq W(\theta_0,\pi)-\gamma  (\sqrt{|\Omega|}+ {1}/{L_{\mu_0}}) = W(\theta_0,\pi)-\epsilon=\max_{\pi'\in \Pi}W(\theta_0,\pi')-\epsilon~.
\end{equation*}
$\pi_T$ is the signal policy required by Lemma~\ref{lemma:existence_for_all_type}, which completes the proof. 
\end{proof}

\subsection{Necessary Condition Proofs}
\label{app:nece_proofs}

\begin{proof}[Proof of Lemma~\ref{prop:continuous_robust_nece}]
Suppose $G(\theta_0)$ is MM-robust.
For any $\epsilon>0$, there exists $\delta>0$ such that for $\Theta=\operatorname{Ball}_{\delta}(\theta_0)$, any max-min policy $\pi_1$ satisfies $ W(\theta,\pi_1) \geq \max_{\pi\in \Pi} W(\theta_0,\pi)-\epsilon$ for every $\theta\in\operatorname{Ball}_{\delta}(\theta_0)$.
For every $\theta\in\operatorname{Ball}_{\delta}(\theta_0)$, $\max_{\pi\in\Pi} W(\theta,\pi)\geq  W(\theta,\pi_1)  \geq \max_{\pi\in\Pi} W(\theta_0,\pi)-\epsilon$. This establishes lower semicontinuity. Combined with the upper semicontinuity (Proposition~\ref{lemma:set_upper_continuous}), $G(\theta_0)$ is continuous.
The proof when $G(\theta_0)$ is RM-robust is the same.
\end{proof}

\begin{lemma}
\label{lemma:lower_dimensional_set_subset}
Let $G(\theta)$ be a BP model.  
For any action $a$, if $B(\theta, a)\neq \emptyset$ and $\dim(B(\theta, a))<|\Omega|-1$, then there exists another action $b$ such that $ B(\theta, a)\subseteq B(\theta,b)$ and $\dim(B(\theta,b))=|\Omega|-1$.
Consequently, $\bigcup_{\{a\in A:\dim(B(\theta,a))=|\Omega|-1\}}B(\theta,a)=\bigcup_{a\in A}B(\theta,a)=\Delta(\Omega)$.
\end{lemma}
\begin{proof}[Proof of Lemma~\ref{lemma:lower_dimensional_set_subset}]
Since $\bigcup_{a\in A}B(\theta,a)=\Delta(\Omega)$ and $\dim(\Delta(\Omega))=|\Omega|-1$, there exists at least one action $b$ such that $\dim(B(\theta,b))=|\Omega|-1$, because a finite union of lower-dimensional polytopes cannot cover \(\Delta(\Omega)\).

Suppose, for contradiction, that for every $b\neq a$ with $\dim(B(\theta,b))=|\Omega|-1$, $B(\theta,a)\setminus B(\theta,b)\neq\emptyset$.
Then for each such $b$, there exists $\mu_b\in B(\theta,a)$ such that $\sum_{\omega} u_{\theta}(a,\omega)\mu_b(\omega)> \sum_{\omega} u_{\theta}(b,\omega)\mu_b(\omega)$.
Let $\mu^*$ be a convex combination of $\{\mu_b:b\neq a, \dim(B(\theta,b))=|\Omega|-1\}$ with strictly positive weights. Since $B(\theta,a)$ is convex, $\mu^*\in B(\theta,a)$.
Then for every $b$ with $\dim(B(\theta,b))=|\Omega|-1$, $\sum_{\omega} u_{\theta}(a,\omega)\mu^*(\omega)> \sum_{\omega} u_{\theta}(b,\omega)\mu^*(\omega)$. 
Hence, by continuity, there is a relative neighborhood of $\mu^*$ in the simplex, denoted by $\operatorname{Ball}_{\gamma}(\mu^*)\bigcap \Delta(\Omega)$, satisfying $\forall\mu\in \operatorname{Ball}_{\gamma}(\mu^*)\bigcap \Delta(\Omega)$, $\sum_{\omega} u_{\theta}(a,\omega)\mu(\omega)> \sum_{\omega} u_{\theta}(b,\omega)\mu(\omega)$ for every $b\neq a$ with $\dim(B(\theta,b))=|\Omega|-1$.

Let $A_a:=\{c\in A: c\neq a, \dim(B(\theta,c))< |\Omega|-1 \}$ denote the collection of other actions with lower-dimensional sets. 
Every belief in $\left(\operatorname{Ball}_{\gamma}(\mu^*)\bigcap \Delta(\Omega)\right)  \setminus\bigcup_{c\in A_a}B(\theta,c)$ belongs to none of the best-reply regions associated with actions other than $a$, and hence belongs to \(B(\theta,a)\).
Note that $\dim(\operatorname{Ball}_{\gamma}(\mu^*)\bigcap \Delta(\Omega))=|\Omega|-1$, and removing finitely many lower-dimensional sets does not reduce its dimension,
\begin{equation*}
 \dim(B(\theta,a)) \geq \dim \left[ \left(\operatorname{Ball}_{\gamma}(\mu^*)\bigcap \Delta(\Omega)\right)\setminus \left(\bigcup_{c\in A_a} B(\theta,c) \right) \right]=|\Omega|-1~.
\end{equation*}
This contradicts the assumption that \(B(\theta,a)\) is lower-dimensional.
\end{proof}

\begin{lemma}
\label{lemma:inf_choice}
Let $G(\theta_0)$ be a BP model.   
If the conditions in Proposition~\ref{prop:suffi_all} fail, then for every optimal policy $\pi$, there exists a critical posterior $\mu\in\operatorname{supp}(\pi)$; furthermore, there exists $b\in A^*(\theta_0,\mu)$ that is not an induced action at $\mu$  and satisfies $\operatorname{dim}(B(\theta_0,b))=|\Omega|-1$, and for every induced action $a$ at $\mu$,  $\sum_{\omega\in \Omega} v(b,\omega) \mu(\omega)<\sum_{\omega\in \Omega} v(a,\omega)\mu(\omega)$.
\end{lemma}

\begin{proof}
Failure of Proposition~\ref{prop:suffi_all}'s condition means every optimal policy supports a critical posterior $\mu$. Let $a$ be a sender-preferred action induced there. If all induced actions have lower-dimensional regions, Lemma~\ref{lemma:lower_dimensional_set_subset} provides $b\in A^*(\theta_0,\mu)$ with $B(\theta_0,a)\subseteq B(\theta_0,b)$ and $\dim B(\theta_0,b)=|\Omega|-1$. Since $b$ is not sender-preferred, its sender payoff is strictly lower. Otherwise, criticality means that some full-dimensional induced $a$ has a duplicate $b\in R(\theta_0,a)$ with strictly lower sender payoff; duplication implies $\dim B(\theta_0,b)=\dim B(\theta_0,a)=|\Omega|-1$.
\end{proof}
\begin{lemma}
 \label{lemma:inf_choice_2}   
For a BP model $G(\theta_0)$, if $\mu$ is critical at $\theta_0$, then $ V(\sigma(\theta_0),\mu)<V(\theta_0,\mu)$.
\end{lemma}
\begin{proof}
At a critical $\mu$, every sender-preferred action in $A^*(\theta_0,\mu)$ is removed by Definition~\ref{def:pseudo_type}, while
$\emptyset\neq A^*(\sigma(\theta_0),\mu)\subseteq A^*(\theta_0,\mu)$. Every remaining action therefore gives strictly less than $V(\theta_0,\mu)$, so $V(\sigma(\theta_0),\mu)<V(\theta_0,\mu)$.
\end{proof}

\begin{definition}
\label{def:extreme_points}
For any convex set $B\subseteq\Delta(\Omega)$ generated by finitely many linear
inequalities, define $\operatorname{extremeP}(B):=\emptyset$, if $B=\emptyset$; otherwise, 
\begin{equation*}
\operatorname{extremeP}(B):=\text{the set of extreme points of }B~.
\label{eq:extreme_points}
\end{equation*}
The notation applies in particular to $\operatorname{extremeP}(B(\theta,a))$ and $\operatorname{extremeP}(B(\sigma(\theta),a))$.
\end{definition}
\begin{definition}
\label{def:basic_policies}
For a type $\theta$ or a virtual type $\sigma(\theta)$, a signal policy $\pi$ is \emph{basic} if:
\begin{enumerate}
    \item its support is contained in $ \bigcup_{a\in A}\operatorname{extremeP}(B(\theta,a))
    $ or $ \bigcup_{a\in A}\operatorname{extremeP}(B(\sigma(\theta),a))$, respectively;
    \item $\pi$ is the unique signal policy whose support is contained in $\operatorname{supp}(\pi)$; equivalently, the probability weights assigned to the posterior beliefs in $\operatorname{supp}(\pi)$ are uniquely determined by Bayes plausibility.
\end{enumerate}
\end{definition}
\begin{lemma}
    \label{lemma:v_pse_upper}
 For every BP model $G(\theta)$, the indirect utility function  for its virtual type $\mu\mapsto V(\sigma(\theta),\mu)$ is upper semicontinuous.
\end{lemma}
\begin{proof}
Each $B(\sigma(\theta),a)$ is compact. If $\mu\notin B(\sigma(\theta),a)$, a neighborhood of $\mu$ is disjoint from that region. Finiteness of $A$ implies that $\mu_n\to\mu$ entails $A^*(\sigma(\theta),\mu_n)\subseteq A^*(\sigma(\theta),\mu)$ eventually. The argument of Lemma~\ref{lemma:V_upper_semi} completes the proof.
\end{proof}

\begin{lemma}
\label{lemma:representation_basic}
For a type $\theta$ or a virtual type $\sigma(\theta)$, there exists an optimal policy that is basic.
\end{lemma}
\begin{proof}
For each $\tau\in\{\theta,\sigma(\theta)\}$, there exists an optimal policy supported on at most $|\Omega|$ posterior beliefs because of the upper semicontinuity of $\mu\mapsto V(\tau,\mu)$ \citep*{Kamenica2011BayesianPersuasion}. Note that each $B(\theta,a)$ and $B(\sigma(\theta),a)$ is convex and the sender's expected utility from a fixed action is linear in the posterior belief. Decompose each supported posterior into extreme points of the relevant best-reply region. Linearity preserves the payoff from the associated action, and sender-preferred tie-breaking can only improve it. Thus, there is an optimal policy $\pi_1$ supported on a finite set of these extreme points.

If Bayes plausibility does not uniquely determine $\pi_1$'s weights, choose another policy $\pi_2\neq\pi_1$ supported on $\operatorname{supp}(\pi_1)$.
Define $\rho:= \max_{\mu\in \operatorname{supp}(\pi_1)} \tfrac{\pi_2(\{\mu\})}{\pi_1(\{\mu\})}$.
Since $\pi_1$ and $\pi_2$ are distinct, $\rho>1$.
Define $\pi_3$ by for every $\mu\in \Delta(\Omega)$, $\pi_3(\{\mu\}):= \tfrac{\rho}{\rho-1} \left[\pi_1 (\{\mu\})- \tfrac{1}{\rho} \pi_2 (\{\mu\})\right]$.
By the definition of $\rho$, for every $\mu\in \operatorname{supp}(\pi_3)\subseteq\operatorname{supp}(\pi_1)$,  $\pi_3(\{\mu\})\geq 0$; the total mass of $\pi_3$ is $1$. The barycenter of $\pi_3$ is $\mu_0$ because $\pi_1$ and $\pi_2$ have the same barycenter $\mu_0$. $\pi_3$ is a signal policy. For $\tau\in\{\theta,\sigma(\theta)\}$
\begin{equation*}
    \int_{\operatorname{supp}(\pi_1)} V(\tau,\mu) \pi_1(d\mu)= \tfrac{1}{\rho}  \int_{ \operatorname{supp}(\pi_1)} V(\tau,\mu) \pi_2(d\mu)+ \tfrac{\rho-1}{\rho} \int_{ \operatorname{supp}(\pi_1)} V(\tau,\mu) \pi_3(d\mu)~.
\end{equation*}
Since $\pi_1$ is optimal, $\pi_2$ and $\pi_3$ must be optimal. At any posterior
attaining the maximum in the definition of $\rho$, $\pi_3$ assigns zero probability. Thus, $\pi_3$ has strictly smaller support than $\pi_1$.
Finite repetition yields an optimal policy whose weights are uniquely determined by its support, hence a basic policy.

\end{proof}

\begin{proof}[Proof of Lemma~\ref{lemma:nece_pseudo_gap}]
Pointwise, $V(\sigma(\theta_0),\mu)\leq V(\theta_0,\mu)$, hence $W(\sigma(\theta_0),\pi)\leq W(\theta_0,\pi)$ for every $\pi$. 
The collection of extreme points of the regions $\{B(\sigma(\theta_0),a)\}_{a\in A}$ is finite; denote it by
$\{\mu_j^{\operatorname{ex}}\}_{j\leq J}$. Since the weights of a basic policy are uniquely determined by its support, there are finitely many basic policies at $\sigma(\theta_0)$ because there are finitely many subsets of extreme points (potential supports). Denote their collection by $\Pi_{\mathrm{basic}}$.
By Lemma~\ref{lemma:representation_basic}, there exists an optimal policy at $\sigma(\theta_0)$ in $\Pi_{\mathrm{basic}}$ and so $\Pi_{\mathrm{basic}}\neq \emptyset$. Let $\Pi_{\theta_0}$ denote the collection of optimal policies at $\theta_0$. 

If $\Pi_{\mathrm{basic}}\setminus\Pi_{\theta_0}=\emptyset$, define $C_1:=\infty$; otherwise, define
\begin{equation*}
    C_1:=\min_{\pi \in \Pi_{\mathrm{basic}} \setminus \Pi_{\theta_0}} \left[  \max_{\pi'\in \Pi} W(\theta_0,\pi') - W(\theta_0,\pi) \right].
\end{equation*}
Since $\Pi_{\mathrm{basic}}\setminus\Pi_{\theta_0}$ is finite and contains no optimal policy at $\theta_0$, $C_1>0$. 
\begin{equation*}
\text{For any $\pi\in \Pi_{\mathrm{basic}} \setminus \Pi_{\theta_0}$:}\quad
W(\sigma(\theta_0),\pi)\leq W(\theta_0,\pi)\leq  \max_{\pi'\in \Pi} W(\theta_0,\pi')-C_1~.
\end{equation*}
We now consider the policies in $\Pi_{\mathrm{basic}}\cap\Pi_{\theta_0}$. 
Note that failure of Proposition~\ref{prop:suffi_all}'s conditions places a critical posterior in every optimal policy's support at $\theta_0$. 
Let $Y:=\left\{\mu_j^{\operatorname{ex}}:\mu_j^{\operatorname{ex}}\text{ is critical at }\theta_0\right\}$.
If $\Pi_{\mathrm{basic}}\cap\Pi_{\theta_0}\neq\emptyset$, then $Y\neq\emptyset$. Define
\begin{equation*}
   q_1:= \min_{\mu\in Y}\left[ V(\theta_0,\mu) - V(\sigma(\theta_0), \mu) \right],\quad \text{and}\quad q_2:=\min_{\pi\in \Pi_{\mathrm{basic}}} \left[\min_{\mu\in \operatorname{supp}(\pi)} \pi(\{\mu\})\right]~.
\end{equation*}
By Lemma~\ref{lemma:inf_choice_2}, the payoff difference in the definition of $q_1$ is strictly positive at every $\mu\in Y$. Since $Y$ is finite, $q_1>0$. Moreover, recalling the definition of basic policies, since $\Pi_{\mathrm{basic}}$ is finite and every posterior in the support of a policy has strictly positive probability, $q_2>0$.
For any $\pi\in \Pi_{\mathrm{basic}}\cap \Pi_{\theta_0}$, choose a critical posterior $\mu_\pi\in\operatorname{supp}(\pi)\cap Y$. Then, 
\begin{equation*}
W(\theta_0,\pi)-W(\sigma(\theta_0),\pi)\geq\pi(\{\mu_\pi\})\left[V(\theta_0,\mu_\pi)-V(\sigma(\theta_0),\mu_\pi)\right]\geq q_1q_2~.
\end{equation*}
If $\Pi_{\mathrm{basic}}\cap \Pi_{\theta_0}=\emptyset$, let $C_2:=\infty$; otherwise, let $C_2:=q_1q_2$.
\begin{equation*}
\text{For any $\pi\in \Pi_{\mathrm{basic}} \cap \Pi_{\theta_0}$:}\quad
W(\sigma(\theta_0),\pi)\leq W(\theta_0,\pi)-C_2\leq  \max_{\pi'\in \Pi} W(\theta_0,\pi')-C_2~.
\end{equation*}
Choosing $C^*= \min\{C_1,C_2\}$ completes the proof.
\end{proof}

\begin{lemma}
\label{lemma:nece_remove_clear}
    Consider the sequence defined in Lemma~\ref{lemma:nece_sequence}. If $B(\theta_0,a)=\emptyset$ or $\dim(B(\theta_0,a))<|\Omega|-1$, then for every $\mu\in \Delta(\Omega)$ and for every $n\in \mathbb{N}^+$, $ a\not\in A^*(\theta_n,\mu)$.
\end{lemma}
\begin{proof}
By Lemma~\ref{lemma:lower_dimensional_set_subset},  $\bigcup_{\{b\in A:\dim(B(\theta_0,b))=|\Omega|-1\}}B(\theta_0,b)=\bigcup_{b\in A}B(\theta_0,b)=\Delta(\Omega)$.
So for every $\mu\in \Delta(\Omega)$, there exists another action $b$ such that $\dim(B(\theta_0,b))=|\Omega|-1$ and $\sum_{\omega\in \Omega} u_{\theta_0}(b,\omega)\mu(\omega)\geq \max_{c\in A} \sum_{\omega\in \Omega} u_{\theta_0}(c,\omega)\mu(\omega)\geq \sum_{\omega\in \Omega} u_{\theta_0}(a,\omega)\mu(\omega)$.
Throughout the sequence in Lemma~\ref{lemma:nece_sequence}, the utility of every action with a full-dimensional best-reply region increases, while the utility of $a$ decreases. So for every $n$, $\sum_{\omega\in\Omega}u_{\theta_n}(b,\omega)\mu(\omega)>\sum_{\omega\in\Omega}u_{\theta_n}(a,\omega)\mu(\omega)$, and therefore $a\notin A^*(\theta_n,\mu)$.
\end{proof}

\begin{proof}[Proof of Lemma~\ref{lemma:nece_sequence}]
Berge's maximum theorem does not apply directly to the sequence-specific comparison.
For any $x\in \mathbb{R}\setminus\{0\}$, define the type $\theta(x)$ by
\begin{equation*}
      u_{\theta(x)}(a,\omega):= \begin{cases}
         u_{\theta_0}(a,\omega) + 2|x| -|x|v(a,\omega) & \text{if } B(\theta_0,a)\neq\emptyset\,,\, \dim(B(\theta_0,a))=|\Omega|-1
         \\
         u_{\theta_0}(a,\omega) - |x| & \text{otherwise} 
     \end{cases}~.
\end{equation*}
 Define function $f:\mathbb{R}\times \Delta(\Omega)\rightarrow \mathbb{R}$ as follows. For any $x$ and $\mu$, if $x=0$, $f(x,\mu):=V(\sigma(\theta_0),\mu)$; otherwise, $f(x,\mu):=V(\theta(x),\mu)$.
For $x\neq0$, the upper semicontinuity of $f$ follows from the upper semicontinuity of $V$ and the continuity of $x\mapsto\theta(x)$. We now verify upper semicontinuity at $x=0$.

As in Lemma~\ref{lemma:nece_remove_clear}, if $B(\theta_0,a)=\emptyset$ or $\dim(B(\theta_0,a))<|\Omega|-1$, then for any $\mu\in \Delta(\Omega)$ and for any $x\neq 0$,  $ a\not\in A^*(\theta(x),\mu)$ because another action strictly dominates it.
Let $A_{\operatorname{res}}:=\{a\in A:  B(\theta_0,a)\neq\emptyset\,,\, \dim(B(\theta_0,a))=|\Omega|-1\}$.
For any $\mu\in\Delta(\Omega)$, recalling that for every $b\in R(\theta_0,a)$, $u_{\theta_0}(a,\omega)=u_{\theta_0}(b,\omega)$ for every $\omega$:
\begin{equation*}
\begin{split}
     A^*(\theta(x),\mu)&=\{a\in A_{\operatorname{res}}: \sum_{\omega\in\Omega}u_{\theta(x)}(a,\omega)\mu(\omega)\geq  \sum_{\omega\in\Omega}u_{\theta(x)}(b,\omega)\mu(\omega),\forall  b\in A_{\operatorname{res}} \}  ~,
     \\
     {}&= \left\{a\in A_{\operatorname{res}}:
\begin{aligned}
&\sum_{\omega\in\Omega} u_{\theta(x)}(a,\omega)\mu(\omega)\geq \sum_{\omega\in\Omega} u_{\theta(x)}(b,\omega)\mu(\omega),&&\forall b \in  A_{\operatorname{res}}\setminus  R(\theta_0,a) 
  \\
&\sum_{\omega\in\Omega} v(a,\omega)\mu(\omega)\leq \sum_{\omega\in\Omega} v(b,\omega)\mu(\omega),&&\forall b\in R(\theta_0,a)
\end{aligned}
\right\}~.   
\end{split}
\end{equation*}  
For any $b\in A\setminus A_{\operatorname{res}}$,  $B(\sigma(\theta_0),b)=\emptyset$. Hence,
\begin{equation*}
\begin{split}
    A^*(\sigma(\theta_0),\mu)&=\{a\in A: \mu\in B(\sigma(\theta_0),a)\}=\{a\in A_{\operatorname{res}}: \mu\in B(\sigma(\theta_0),a)\}~,
    \\
    {}&= \left\{a\in A_{\operatorname{res}}:
\begin{aligned}
&\sum_{\omega\in\Omega} u_{\theta_0}(a,\omega)\mu(\omega)\geq \sum_{\omega\in\Omega} u_{\theta_0}(b,\omega)\mu(\omega),&&\forall b \in  A_{\operatorname{res}}\setminus  R(\theta_0,a) 
  \\
&\sum_{\omega\in\Omega} v(a,\omega)\mu(\omega)\leq \sum_{\omega\in\Omega} v(b,\omega)\mu(\omega), &&\forall b\in R(\theta_0,a) 
\end{aligned}
\right\}~.   
\end{split}
\end{equation*}
Thus, by the same argument as in Lemma~\ref{lemma:V_upper_semi},  for every posterior belief $\mu$, for every sequence $x_n\rightarrow 0$ and $\mu_n\rightarrow\mu$, $A^*(\theta(x_n),\mu_n)\subseteq A^*(\sigma(\theta_0),\mu)$ eventually, and so $f$ is upper semicontinuous. 

For any $x$ and $\pi$, define $F:\mathbb{R}\times \Pi \rightarrow \mathbb{R}$ by $ F(x,\pi)=\int_{\Delta(\Omega)} f(x,\mu)\pi(d\mu)$.
By the same proof as for Lemma~\ref{lemma:joint_upper_W} (replacing $\theta$ with $x$), $F$ is upper semicontinuous on $\mathbb{R}\times \Pi$. 
Thus, applying the upper semicontinuity part of Berge's maximum theorem (same as the proof of Proposition~\ref{lemma:set_upper_continuous}, replacing $W$ with $F$ and $\theta$ with $x$), $\max_{\pi\in \Pi} F(x,\pi)$ is upper semicontinuous with respect to $x$. 
Consider the sequence $\{\tfrac{1}{n}\}_{n\in \mathbb{N}^+}$.
Recalling the construction, for every $n\in \mathbb{N}^+$, $f(\tfrac{1}{n},\mu)=V(\theta_n,\mu)$ and $F(\tfrac{1}{n}, \pi)=W(\theta_n,\pi)$; $f(0,\mu)=V(\sigma(\theta_0),\mu)$ and $F(0, \pi)=W(\sigma(\theta_0),\pi)$. 
Thus, 
\begin{equation*}
    \limsup_{n\rightarrow \infty} \left[ \max_{\pi\in \Pi} W(\theta_n,\pi)\right] =\limsup_{n\rightarrow \infty} \left[ \max_{\pi\in \Pi} F(\tfrac{1}{n},\pi)\right] \leq \max_{\pi\in \Pi} F(0,\pi)= \max_{\pi\in \Pi} W(\sigma(\theta_0),\pi)~.
\end{equation*}
\end{proof}

\renewcommand{\theequation}{C.\arabic{equation}}
\renewcommand{\theHequation}{C.\arabic{equation}}
\setcounter{equation}{0}

\section{Proofs for Section \ref{sec:tie_eq}}
\label{app:tie_eq}

\begin{lemma}
\label{lemma:mapping_tie_1}
Let $G(\theta_0)$ be a BP model. For any $\gamma>0$, there exists a Borel measurable mapping $T:\Delta(\Omega)\to\Delta(\Omega)$ such that, if $\mu$ is not critical at $\theta_0$, then $||\mu-T(\mu)||_2\leq \gamma$ and $\min_{a\in A^*(\theta_0,T(\mu))}  \left(\sum_{\omega\in \Omega} T(\mu)(\omega) v(a,\omega) \right)  \geq V(\theta_0,\mu)-\sqrt{|\Omega|}\gamma$.
\end{lemma}
\begin{proof}
For each action $a$ with $\dim(B(\theta_0,a))=|\Omega|-1$, let $\mu_a$ be a posterior in its relative interior. 
Define the mapping $T$ as follows. For every $\mu\in \Delta(\Omega)$
\begin{enumerate}
    \item if $\mu$ is critical at $\theta_0$, let $T(\mu):=\mu$.
     \item otherwise, choose by a fixed ordering a sender-optimal $a\in A^*(\theta_0,\mu)$ satisfying the two conditions in Definition~\ref{def:critical_mu}, then if $\mu$ belongs to the relative interior of $B(\theta_0,a)$, let $ T(\mu):=  \mu$; otherwise let 
\begin{equation*}
    T(\mu):= \left(\tfrac{\gamma}{||\mu-\mu_a||_2} \mu_a+\tfrac{||\mu-\mu_a||_2-\gamma}{||\mu-\mu_a||_2} \mu \right)\mathbbm{1}_{||\mu_a-\mu||_2\geq \gamma}+ \mu_a(1-\mathbbm{1}_{||\mu_a-\mu||_2\geq \gamma})~.
\end{equation*}
    Thus, $||T(\mu)-\mu||_2\leq \gamma$, and $A^*(\theta_0,T(\mu))=R(\theta_0,a)$.
\end{enumerate}
By an argument similar to that in the proof of Lemma~\ref{lemma:intersect_small_set}, this mapping is Borel measurable and satisfies $\min_{b\in R(\theta_0,a)} \left[ \sum_{\omega\in\Omega} v(b,\omega)T(\mu)(\omega)  \right]\geq V(\theta_0,\mu) -\sqrt{|\Omega|}\gamma$ for every non-critical $\mu$ at $\theta_0$.
\end{proof}
\begin{lemma}
\label{lemma:sufficient_tie}
\label{prop:suffi_tie}
Let $G(\theta_0)$ be a BP model.  
If there exists an optimal signal policy $\pi$ such that all the posterior beliefs in its support are not critical at $\theta_0$, then $G(\theta_0)$ is robust to tie-breaking rules.
\end{lemma}
\begin{proof}
Consider a sequence $\{\gamma_n=\tfrac{1}{n}\}_{n\in\mathbb{N}^+}$. Let $\pi$ be the optimal signal policy that satisfies the above conditions. For each $n$, applying Lemma~\ref{lemma:mapping_tie_1} and the technical device Lemma~\ref{lemma:adjustment_policy_defi} in Appendix~\ref{app:alt_proofs_geo}, the signal policy $\pi_{T_n}$ satisfies:
    \begin{equation*}
\begin{split}
\int_{\Delta(\Omega)}  \min_{a\in A^*(\theta_0,\mu)}  \left(\sum_{\omega\in \Omega} \mu(\omega) v(a,\omega) \right)  \pi_{T_n}(d\mu) & \geq  W(
\theta_0,\pi) -\sqrt{|\Omega|}\gamma_n-\frac{\gamma_n}{L_{\mu_0}}~.
\end{split}
\end{equation*}
The detailed derivation is similar to that in Corollary~\ref{coro:coro_adjust}. Since $\pi$ is optimal at $\theta_0$,  $ \liminf_{n\rightarrow \infty} \left[\int_{\Delta(\Omega)}  \min_{a\in A^*(\theta_0,\mu)}  \left(\sum_{\omega\in \Omega} \mu(\omega) v(a,\omega) \right)  \pi_{T_n}(d\mu)\right]\geq \max_{\pi'\in \Pi} W(\theta_0,\pi')$.
Thus,  $ \sup_{\pi'\in\Pi}\int_{\Delta(\Omega)}\min_{a\in A^*(\theta_0,\mu)}\left(\sum_{\omega\in\Omega}\mu(\omega)v(a,\omega)\right)\pi'(d\mu)\geq \max_{\pi'\in \Pi} W(\theta_0,\pi')$.
The reverse inequality follows immediately because sender-optimal tie-breaking weakly dominates sender-worst tie-breaking for every signal policy.
\end{proof}

\begin{proof}[Proof of Proposition~\ref{prop:nece_suf_tie_breaking}]
The conditions in Proposition~\ref{prop:suffi_all} are necessary and sufficient for robustness to receiver preferences. Lemma~\ref{prop:suffi_tie} shows that these conditions are sufficient for robustness to tie-breaking rules, whereas Lemmas~\ref{lemma:nece_pseudo_gap} and~\ref{lemma:eq_tie_pseudo} show that they are necessary.
\end{proof}

\begin{proof}[Proof of Lemma~\ref{lemma:tight_tie_breaking}]
    For any $\mu$, the action $a$ selected under $\Sigma$ satisfies  $\dim(B(\theta,a))=|\Omega|-1$ and $\min_{b\in R(\theta,a)} \sum_{\omega\in\Omega}v(b,\omega)\mu(\omega)= \sum_{\omega\in\Omega}v(a,\omega)\mu(\omega)$. 
   As in the proof of  Lemma~\ref{lemma:mapping_tie_1}, choose $T(\mu)$ in the relative interior of $B(\theta,a)$ such that $\|\mu-T(\mu)\|_2\leq\gamma$ and $A^*(\theta,T(\mu))=R(\theta,a)$.
 We then claim that, by an argument similar to that in the proof of Lemma~\ref{lemma:intersect_small_set}, for any $\gamma>0$, there exists a Borel measurable mapping $T:\Delta(\Omega)\rightarrow\Delta(\Omega)$ such that for any $\mu\in \Delta(\Omega)$, $||T(\mu)-\mu||_2\leq \gamma$ and $\min_{a\in A^*(\theta,T(\mu))} \sum_{\omega\in \Omega} v(a,\omega)T(\mu)(\omega) \geq V_{\Sigma}(\theta,\mu)-\sqrt{|\Omega|}\gamma$.
 
    By the same argument applied in the proof of Lemma~\ref{lemma:intersect_small_set}, $ \min_{b\in R(\theta,a)} \left[ \sum_{\omega\in\Omega} v(b,\omega)T(\mu)(\omega)  \right]\geq V_{\Sigma}(\theta,\mu) -\sqrt{|\Omega|}\gamma$.
   The remainder of the proof is analogous to that of Lemma~\ref{prop:suffi_tie}: applying Lemma~\ref{lemma:adjustment_policy_defi}  shows that the sender-worst tie-breaking value can approximate
$\max_{\pi\in\Pi}\int V_{\Sigma}(\theta,\mu)\pi(d\mu)$ arbitrarily closely with the error $\sqrt{|\Omega|}\gamma+\tfrac{\gamma}{L_{\mu_0}}\rightarrow 0$ as $\gamma\rightarrow0$.
The reverse inequality follows immediately because tie-breaking $\Sigma$ weakly dominates sender-worst tie-breaking for every signal policy.
\end{proof}

\begin{example}
\label{example:appendix_noUniform}
Let $\theta_0$ index the BP model in Example~\ref{sec:warm_examples}. Consider a sequence of robust BP models indexed by $\theta_n$, where $n\in\mathbb{N}^+$ and $n\geq 20$, convergent to $\theta_0$ in receiver utilities. The receiver's utilities are shown below. For each $n$, the expected utility under the corresponding optimal policy exceeds $0.5$. Fix $\epsilon=0.01$. For each $n$, let $\delta_n$ be a radius such that for every $\theta\in\operatorname{Ball}_{\delta_n}(\theta_n)$, $ | \max_{\pi\in \Pi} W(\theta,\pi)- \max_{\pi\in \Pi} W(\theta_n,\pi)|<\epsilon$.
\begin{figure}[h!]
  \centering
 \begin{subfigure}[c]{0.45\textwidth}
      \begin{tabular}{c|c|c|c}
\hline
$u_{\theta_0}(a,\omega)$ & \textbf{$a_1$} & { \textbf{$a_2$}}& \textbf{$a_3$}\\ \hline
$\omega=0$ & 0.84 & $ 0.48 $& 0.05   \\ \hline
${\omega}=1$ & 0.72 & 0.84 & $0.84$   \\ \hline
\end{tabular}
    \caption{$\theta_0$-type receiver utility }
    \label{table:sequence_discu2}
    \end{subfigure}
     \begin{subfigure}[c]{0.45\linewidth}
   \begin{tabular}{c|c|c|c}
\hline
$u_{\theta_n}(a,\omega)$ & \textbf{$a_1$} & { \textbf{$a_2$}}& \textbf{$a_3$}\\ \hline
$\omega=0$ & 0.84 & $ 0.48 $& 0.05   \\ \hline
${\omega}=1$ & 0.72 & 0.84 & $0.84+  \tfrac{1}{n}$   \\ \hline
\end{tabular}
    \caption{$\theta_n$-type receiver utility }
    \label{table:sequence_discu1}
  \end{subfigure}\hfill
\end{figure}

If $\delta_n>\tfrac{1}{n}$, then $\theta_0$ lies in the interior of $\operatorname{Ball}_{\delta_n}(\theta_n)$. Hence, there exists an instance $\theta'\in \operatorname{Ball}_{\delta_n}(\theta_n)$ such that $u_{\theta'}(a_3,1)<0.84$ and $ u_{\theta'}(a,\omega)=  u_{\theta_0}(a,\omega) $ for every   $(a,\omega)\neq (a_3,1)$.
By the analysis in Example~\ref{sec:warm_examples}, $ \max_{\pi\in \Pi} W(\theta',\pi)<0.1$.
Since $\max_{\pi\in\Pi}W(\theta_n,\pi)>0.5$, it follows that $|\max_{\pi\in\Pi}W(\theta',\pi)-\max_{\pi\in\Pi}W(\theta_n,\pi)|>0.4>\epsilon$, a contradiction.
Therefore, for every $n$, $\delta_n\leq \tfrac{1}{n} \rightarrow 0$ as $n\rightarrow \infty$. 
\end{example}

\renewcommand{\theequation}{D.\arabic{equation}}
\renewcommand{\theHequation}{D.\arabic{equation}}
\setcounter{equation}{0}
\section{Technical Device and Alternative Proofs}
\label{app:alt_proofs_geo}
Similar constructions appear in, for example,  \citet[proof of Proposition 6]{Kamenica2011BayesianPersuasion}, \citet[proof of Lemma 5]{arieli2023optimal}, \citet[proof of Theorem 1]{lipnowski2024perfect}, and \citet[proof of Lemma C.1]{lin2025generalized}. The latter two use relevant constructions to address robustness questions. The technical difference from them is that we avoid relying on a strictly positive inducibility gap and avoid purely analytical limit arguments.
\begin{lemma}
\label{lemma:far_point_around_u0}
Let $L_{\mu_0}=\min_{\omega\in\Omega} \mu_0(\omega)$ and let $x\in \mathbb{R}^{|\Omega|}$. If $||x-\mu_0||_2\leq L_{\mu_0}$ and if  $\sum_{\omega\in \Omega} x(\omega)=1$, then $x\in\Delta(\Omega)$.
\end{lemma}
\begin{proof}
 $x(\omega)\geq \mu_0(\omega)-L_{\mu_0}\geq 0$  for every $\omega$. So $x\in \Delta(\Omega)$.
\end{proof}

\begin{lemma}
   \label{lemma:adjustment_policy_defi} 
 Let $\pi$ be a signal policy and let $T:\Delta(\Omega)\rightarrow\Delta(\Omega)$ be a Borel measurable map.
Define the {\em adjustment} of $\pi$ under $T$, ${\pi}_{T}$,  by
\begin{equation*}
   {\pi}_{T} (d\mu) =\frac{L_{\mu_0}}{L_{\mu_0}+||r||_2} \int_{x
\in\Delta(\Omega)} \operatorname{Dirac}_{T(x)}(d\mu)\pi(dx) +\frac{||r||_2}{L_{\mu_0}+||r||_2} \operatorname{Dirac}_{\mu_{\operatorname{corr}}}(d\mu)~,
\end{equation*}
where $r:=\int_{\Delta(\Omega)} T(\mu)\pi(d\mu)-\mu_0$; if $r\neq 0$, $\mu_{\operatorname{corr}}:=\mu_0-\frac{L_{\mu_0}}{||r||_2} r$, otherwise, $\mu_{\operatorname{corr}}:=\mu_0$.
Then $\pi_T$ is a signal policy.
\end{lemma}

\begin{proof}
If $r=0$, it immediately follows that $\pi_T$ is a signal policy.
If $r\neq 0$, we first show that $\mu_{\mathrm{corr}}\in\Delta(\Omega)$. Specifically, $||\mu_{\operatorname{corr}}-\mu_0||_2=L_{\mu_0}$, and
\begin{equation*}
\begin{split}
     \sum_{\omega\in \Omega} \mu_{\operatorname{corr}}(\omega)&= \sum_{\omega\in \Omega} \mu_{0}(\omega)-  \tfrac{L_{\mu_0}}{||r||_2} \sum_{\omega\in \Omega} r(\omega)=1- \tfrac{L_{\mu_0}}{||r||_2} \sum_{\omega\in\Omega}\left(\int_{\Delta(\Omega)} T(\mu)(\omega)\pi(d\mu)-\mu_0(\omega)\right)
     \\
      {}  &=1-\tfrac{L_{\mu_0}}{||r||_2} \left[\int_{\Delta(\Omega)} \left( \sum_{\omega\in\Omega}T(\mu)(\omega)\right)\pi(d\mu)-\left( \sum_{\omega\in\Omega}\mu_0(\omega)\right)\right]=1~.
\end{split}
\end{equation*}
We show that $\pi_T$ is a probability measure and verify Bayes plausibility: 
\begin{equation*}
\pi_T(\Delta(\Omega))=\frac{L_{\mu_0}}{L_{\mu_0}+||r||_2}\int_{x\in\Delta(\Omega)}\pi(dx)+\frac{||r||_2}{L_{\mu_0}+||r||_2}=1 ~.    
\end{equation*}
\begin{equation*}
\begin{split}
\int_{\Delta(\Omega)} \mu \pi_T(d\mu)
        &=
        \frac{L_{\mu_0}}{L_{\mu_0}+||r||_2}
        \int_{\Delta(\Omega)}
        \int_{\Delta(\Omega)}
        \mu \mathrm{Dirac}_{T(x)}(d\mu)\pi(dx)
        +
        \frac{||r||_2}{L_{\mu_0}+||r||_2}
        \mu_{\mathrm{corr}} 
        \\
        {}&= \frac{L_{\mu_0}}{L_{\mu_0}+||r||_2}
        \int_{\Delta(\Omega)}
         T(x)\pi(dx)
        +
        \frac{||r||_2}{L_{\mu_0}+||r||_2}
        \mu_{\mathrm{corr}} 
        \\
        {}&= \frac{L_{\mu_0}}{L_{\mu_0}+||r||_2}
        (\mu_0+r)
        +
        \frac{||r||_2}{L_{\mu_0}+||r||_2}
        \left(
        \mu_0-\frac{L_{\mu_0}}{||r||_2}r
        \right)=\mu_0~.
\end{split}
\end{equation*}
\end{proof}
\begin{corollary}
\label{coro:coro_adjust}
Fix two types, $\theta_1$ and $\theta_2$, some $\gamma>0$ and a signal policy $\pi$. 
If for every $\mu\in \operatorname{supp}(\pi)$, the mapping $T:\Delta(\Omega)\rightarrow\Delta(\Omega)$ satisfies $ ||\mu-T(\mu)||_2\leq \gamma$ and $V(\theta_2,T(\mu))\geq V(\theta_1,\mu)-\sqrt{|\Omega|}\gamma$,
then $\pi_T$ defined in Lemma~\ref{lemma:adjustment_policy_defi} satisfies $ W(\theta_2,\pi_T)\geq W(\theta_1,\pi)  -\sqrt{|\Omega|}\gamma-\frac{\gamma}{L_{\mu_0}}$.
\end{corollary}
\begin{proof} 
Note that $||r||_2\leq \int_{\Delta(\Omega)} ||\mu-T(\mu)||_2\pi(d\mu)\leq \gamma$, $0\leq V(\theta,\mu)\leq 1$ for every $(\theta,\mu)$, and $0\leq W(\theta,\pi)\leq1$ for every $(\theta,\pi)$.
So $\pi_T$ satisfies
\begin{equation*}
\begin{split}
      W(\theta_2,\pi_T)&\geq \frac{L_{\mu_0}}{L_{\mu_0}+||r||_2} \int_{x\in\operatorname{supp}(\pi)} V(\theta_2,T(x))\pi(dx)
\\
{}&\geq \frac{L_{\mu_0}}{L_{\mu_0}+||r||_2} \int_{x\in\operatorname{supp}(\pi)}\left( V(\theta_1,x)- \sqrt{|\Omega|}\gamma\right) \pi(dx)
\\
{}&\geq W(\theta_1,\pi) -\frac{||r||_2}{L_{\mu_0}+||r||_2}- \frac{L_{\mu_0}}{L_{\mu_0}+||r||_2} \sqrt{|\Omega|}\gamma
\geq W(\theta_1,\pi) -\sqrt{|\Omega|}\gamma-\frac{\gamma}{L_{\mu_0}}~.
\end{split}
\end{equation*}
\end{proof}

\subsection{Alternative Proofs for Proposition~\ref{lemma:set_upper_continuous} and Lemma~\ref{lemma:nece_sequence}}
\label{app:alter_detailed_proofs}
\begin{lemma}
\label{lemma:upper_sets_appendix}
Let $G(\theta_0)$ be a BP model. For any $\gamma>0$, there exists $\delta>0$ such that for every $\theta\in \operatorname{Ball}_{\delta}(\theta_0)$ and for every action $a$, (1) $B(\theta,a)\neq \emptyset$ only if $B(\theta_0,a)\neq \emptyset$; (2) if $B(\theta,a)\neq \emptyset$, then $\max_{x\in B(\theta,a)} \min_{y\in B(\theta_0,a)} ||x-y||_2\leq \gamma$.
\end{lemma}
\begin{proof}
    Consider the function $f(\theta,a)=\max_{\mu\in\Delta(\Omega)} \min_{b\in A, b\ne a} \left[ \sum_{\omega\in \Omega} (u_{\theta}(a,\omega)-u_{\theta}(b,\omega) )\mu(\omega)  \right]$ for each $a$. $f$ is continuous in $\theta$. For every $\theta$, $B(\theta,a)=\emptyset$ if and only if $f(\theta,a)<0$. 
   Suppose that $B(\theta_0,a)=\emptyset$. Then, by continuity, there exists $\delta_a>0$ such that $f(\theta,a)<0$ whenever $\theta\in\operatorname{Ball}_{\delta_a}(\theta_0)$. Because $A$ is finite, we choose the smallest required $\delta$, which completes the proof of the first statement.

Each $B(\theta,a)$ is defined by finitely many weak linear inequalities with coefficients continuous in $\theta$. Thus, for every $\theta_n\to\theta$ and $\mu_n\to\mu$ with $\mu_n\in B(\theta_n,a)$, $\mu\in B(\theta,a)$: the correspondence has a closed graph, hence is upper hemicontinuous \citep[Chapter 5, Theorem 5.7]{rockafellar1998variational}. Definition~\ref{def:set_convergence} and the finiteness of $A$ complete the proof of the second statement.
\end{proof}

The following direct corollary of the previous lemma---parallel to Lemma~\ref{lemma:intersect_small_set}, selecting a sender-optimal induced action by the fixed ordering and projecting onto its baseline best-reply region---and Corollary~\ref{coro:coro_adjust} complete the proof of Proposition~\ref{lemma:set_upper_continuous}.

\begin{corollary}
Let $G(\theta_0)$ be an arbitrary BP model. For any $\gamma>0$, there exists $\delta>0$ such that for any 
$\theta\in \operatorname{Ball}_{\delta}(\theta_0)$, there exists a Borel measurable mapping $T:\Delta(\Omega)\to\Delta(\Omega)$ satisfying for every $\mu\in\Delta(\Omega)$, $||\mu-T(\mu)||_2\leq \gamma$ and $V(\theta_0,T(\mu))\geq V(\theta,\mu)-\sqrt{|\Omega|}\gamma$.
\end{corollary}

To obtain a similar argument that completes the proof of Lemma~\ref{lemma:nece_sequence}, note that
\begin{lemma}
    \label{lemma:removal_lower_empty}
    Let $G(\theta)$ be an arbitrary BP model. Define $ A_{\operatorname{remain}}(\theta):=\{a\in A: B(\theta,a)\neq \emptyset \text{ and } \dim(B(\theta,a))=|\Omega|-1\}$.
Then for every $a\in  A$ and for any $A'$ such that $ A_{\operatorname{remain}}(\theta)\subseteq A'\subseteq A$,  $ B(\theta,a)=\{\mu\in \Delta(\Omega): \sum_{\omega\in\Omega} u_{\theta}(a,\omega)\mu(\omega)\geq \sum_{\omega\in\Omega} u_{\theta}(b,\omega)\mu(\omega), \forall b\in  A' \}$.
\end{lemma}
\begin{proof}
The left-hand side is contained in the displayed right-hand side because $A'\subseteq A$. Conversely, if $\mu\notin B(\theta,a)$, some best reply $c$ at $\mu$ strictly dominates $a$. Either $B(\theta,c)$ is full-dimensional or, by Lemma~\ref{lemma:lower_dimensional_set_subset}, it is contained in a full-dimensional $B(\theta,b)$. In either case, some $b\in A_{\operatorname{remain}}(\theta)\subseteq A'$ strictly dominates $a$ at $\mu$, so $\mu$ also lies outside the right-hand side.
\end{proof}
For the virtual type (Definition~\ref{def:pseudo_type}) and the sequence (Lemma~\ref{lemma:nece_sequence}) in Section~\ref{sec:nece_example}:
\begin{lemma}
\label{lemma:nece_corre_alter}
   Let $G(\theta_0)$ be a BP model. For any $\gamma>0$, there exists $N$ such that for every $n\geq N$ and for every action $a$, (1) $B(\theta_n,a)\neq \emptyset$ only if $B(\sigma(\theta_0),a)\neq \emptyset$; (2) if $B(\theta_n,a)\neq \emptyset$, then $\max_{x\in B(\theta_n,a)} \min_{y\in B(\sigma(\theta_0),a)} ||x-y||_2\leq \gamma$.
\end{lemma}
\begin{proof}
    For any $a\in A\setminus  A_{\operatorname{remain}}(\theta_0)$, $B(\sigma(\theta_0),a)=\emptyset$.
    Recall from the proof of Lemma~\ref{lemma:nece_sequence} that, for any $n$ and for any $a\in A\setminus  A_{\operatorname{remain}}(\theta_0)$, $B(\theta_n,a)=\emptyset$, so $A_{\operatorname{remain}}(\theta_n) \subseteq  A_{\operatorname{remain}}(\theta_0)$.

    For every $a\in  A_{\operatorname{remain}}(\theta_0)$ and for every $b\in R(\theta_0,a)$, $u_{\theta_0}(a,\omega)=u_{\theta_0}(b,\omega)$ for every $\omega$. Thus, the corresponding inequality is simply $0\geq0$ and does not contribute to the definition of the set.
    By Definition~\ref{def:pseudo_type} and applying the previous lemma:
    \begin{equation*}
\begin{split}
    B(\sigma(\theta_0), a)&=\{\mu\in B(\theta_0,a): \sum_{\omega\in\Omega} v(a,\omega)\mu(\omega)\leq \sum_{\omega\in\Omega} v(b,\omega)\mu(\omega), \quad \forall b \in   R(\theta_0,a)\}~,
    \\
    {}&= \left\{\mu\in \Delta(\Omega):
\begin{aligned}
&\sum_{\omega\in\Omega} u_{\theta_0}(a,\omega)\mu(\omega)\geq \sum_{\omega\in\Omega} u_{\theta_0}(b,\omega)\mu(\omega),&&\forall b \in  A_{\operatorname{remain}}(\theta_0)\setminus  R(\theta_0,a) 
  \\
&\sum_{\omega\in\Omega} v(a,\omega)\mu(\omega)\leq \sum_{\omega\in\Omega} v(b,\omega)\mu(\omega), &&\forall b\in R(\theta_0,a) 
\end{aligned}
\right\}.   
\end{split}
\end{equation*}
Along the sequence, because $A_{\operatorname{remain}}(\theta_n)\subseteq A_{\operatorname{remain}}(\theta_0)$,
\begin{equation*}
\begin{split}
    B(\theta_n, a)&=\{\mu\in \Delta(\Omega): \sum_{\omega\in\Omega} u_{\theta_n}(a,\omega)\mu(\omega)\geq \sum_{\omega\in\Omega} u_{\theta_n}(b,\omega)\mu(\omega), \quad \forall b \in A\}~,
    \\
    &=\{\mu\in \Delta(\Omega): \sum_{\omega\in\Omega} u_{\theta_n}(a,\omega)\mu(\omega)\geq \sum_{\omega\in\Omega} u_{\theta_n}(b,\omega)\mu(\omega), \quad \forall b \in A_{\operatorname{remain}}(\theta_0)\}~,
    \\
    {}&= \left\{\mu\in \Delta(\Omega):
\begin{aligned}
&\sum_{\omega\in\Omega} u_{\theta_n}(a,\omega)\mu(\omega)\geq \sum_{\omega\in\Omega} u_{\theta_n}(b,\omega)\mu(\omega),&&\forall b \in  A_{\operatorname{remain}}(\theta_0)\setminus  R(\theta_0,a) 
  \\
&\sum_{\omega\in\Omega} v(a,\omega)\mu(\omega)\leq \sum_{\omega\in\Omega} v(b,\omega)\mu(\omega), &&\forall b\in R(\theta_0,a) 
\end{aligned}
\right\}.   
\end{split}
\end{equation*}
The remainder of the argument follows the same steps as the alternative proof of Proposition~\ref{lemma:set_upper_continuous}. First, the set convergence result $\limsup_{n\to\infty} B(\theta_n,a)\subseteq B(\sigma(\theta_0),a)$ follows by the same reasoning as in the proof of Lemma~\ref{lemma:upper_sets_appendix}. Then, we can construct the measurable mapping. Corollary~\ref{coro:coro_adjust} then completes the proof.
\end{proof}

\bibliographystyle{ims}
\bibliography{reference}

\end{document}